\documentclass[final, 3p, times]{elsarticle}

\usepackage{latexsym, amsmath, amssymb,amsthm, graphicx}
\usepackage[bookmarks, a4paper, colorlinks=true]{hyperref}

\newtheorem{theorem}{Theorem}
\newtheorem{lemma}{Lemma}
\newtheorem{corollary}{Corollary}
\newtheorem{remark}{Remark}

\begin{document}

\begin{frontmatter}

\title{Analysis on the Pricing model for a Discrete Coupon Bond with Early redemption provision by the Structural Approach}
\author{Hyong-Chol O ${}^{1}$, Tae-Song Kim${}^{2}$}
\ead{hc.o@ryongnamsan.edu.kp}
\address{Faculty of Mathematics, \textbf{Kim Il Sung} University, Pyongyang, D.P.R.Korea}

\begin{abstract}
In this paper, using the structural approach is derived a mathematical model of the discrete coupon bond with the provision that allows the holder to demand early redemption at any coupon dates prior to the maturity and based on this model is provided some analysis including min-max and gradient estimates of the bond price. Using these estimates the existence and uniqueness of the default boundaries and some relationships between the design parameters of the discrete coupon bond with early redemption provision are described. Then under some assumptions the existence and uniqueness of the early redemption boundaries is proved and the analytic formula of the bond price is provided using higher binary options. Finally for our bond is provided the analysis on the duration and credit spread, which are used widely in financial reality. Our works provide a design guide of the discrete coupon bond with the early redemption provision.
\end{abstract}

\end{frontmatter}

\section{Introduction}
Issueing bond is a kind of financing methods of firms and among firm bonds there is a bond with the provision under which the bond holder can demand early redemption prior to the maturity. This is a kind of defaultable corporate bonds and the pricing problem of defaultable corporate bonds is one of the most promising areas in financial mathematics \cite{Agl}.

It is well known that there are two main approaches among methods to price defaultable corporate bonds: one is the structural approach and another one is the reduced form approach. In the structural approach, it is thought that the default events occur when the firm value is not sufficient to repay debt, that is, the firm value attains a certain lower threshold(\textit{default barrier }or\textit{ default boundary}) from the above \cite{{OKRJ},{LS}}. In the reduced form approach, they think that it is possible for the default event to  occur at any time and the default event is an unpredictable event without any relation to the firm value. In the reduced form approach, if the default probability in time interval $[t,\; t+\Delta t]$ is $\lambda \Delta t$, then $\lambda $ is called \textit{default intensity }or \textit{hazard rate} \cite{{DS},{Jar}}. The third approach is to unify the structural and reduced form approaches \cite{{BP},{CL},{MU1},{MU2}}. As for the history of the above approaches and their advantages and shortcomings, readers can refer to the introductions of \cite{{BP},{CL}}.

The related information such as default barrier and default intensity is related to the internal business information of companies and the structural and reduced form approaches can be used to design the corporate bonds.

In reality it is very hard for investors out of the company to get the information of the company in the whole life time interval of the bond. They might probably know only the discretely (for example, every year or every three months etc.) declared informations. Hence the modelling of corporate bonds using only the discrete default information was proposed with the purpose of making the study of credit risk close to the financial reality. In this direction, \cite{{OKRJ},{OKK}} gives some results of zero coupon bonds using higher binary options(\cite{OK}).

There have been many studies of theoretical modelling of the price of zero coupon bonds which are originated in \cite{Mer}, whereas studies of realistic payout structure providing fixed discrete coupons are relatively less\cite{Agl}.  Geske\cite{Ges} studied this problem at first, where the author models discrete interest payouts prior to maturity as determinants of default risk. Agliardi\cite{Agl} generalized the formula of \cite{Ges} for defaultable coupon bonds and studied a stochastic risk free term structure and the effects of bankruptcy cost and government taxes on bond interest and calculated the duration of defaultable bonds. Agliardi's approach in \cite{Agl} is a kind of structural approach.

\cite{{OKP},{OJKJ}} studied the problem of generalizing the structural model of \cite{Agl} into the comprehensive unified model of structural and reduced form approaches. \cite{OKP} obtained the pricing formula of the corporate bond with discounted discrete coupon in unified two-factor model of structural and reduced form approaches. \cite{OJKJ} obtained the pricing formula of the corporate bond with fixed discrete coupon in unified one-factor model of structural and reduced form approaches.

In \cite{{Agl},{OKP},{OJKJ}}, they studied the discrete coupon bonds without the early redemption provision. However, many firms issue and use discrete coupon bonds with the provision that allow the holder to demand early redemption prior to the maturity.

Generally speaking, discrete coupon bonds with the provision that allow the holder to demand early redemption prior to the maturity are included in the class of puttable bonds (or bond options)\cite{{BM},{Hull}} and widely used in many companies but it seems difficult to find studies on the their concrete pricing models and price estimates, there are only some works on general pricing bond option on zero coupon bonds\cite{{Jam},{ZJEJ},{ZS}}. 

In this paper we derive a PDE model for the price of a discrete coupon bond with the provision that allows the holder to demand early redemption at any coupon dates prior to the maturity, and based on the financial analysis on the relationships between the design parameters of the bond we prove the existence and uniqueness of the default and early redemption boundaries. Then we give the analytic pricing formula of the bond using higher binary options, and some applications including the analysis on duration and credit spread. Our works provide some design guide of the discrete coupon bond with the 
early redemption provision.

The remainder of the article is organized as follows. In Section 2 we consider the pricing model of a discrete coupon bond with the early redemption provision. In Section 3 we prove the existence and uniqueness of the default boundaries and describe some relationships between the design parameters of the discrete coupon bond with early redemption provision. In Section 4 we prove the existence and uniqueness of the early redemption boundaries and Section 5 gives the analytic pricing formula of the bond. Section 6 provides some applications including the analysis on duration and credit spread.

\section{The mathematical model of the bond price}
\textbf{Assumptions}

1) The short rate $r$ is a constant. Then the price of default free zero coupon bond with maturity $T$ and face value 1 is
\[Z(t;\;T)=e^{-r(T-t)} .\]

2) The firm value process $V(t)$ follows a geometric Brownian motion
\[dV(t)=(r-b)V(t)dt+s_{V} V(t)dW(t)\]
under the risk neutral martingale measure. Here $b\ge 0$ is a constant.

3) Let $0=T_{0} <T_{1} <\cdots <T_{N-1} <T_{N} =T$ and let $T$ be the maturity of our corporate bond with face value $F$(unit of currency). At time $T_{i} \; (i=1,\cdots , N-1)$, the bond holder receives the coupon of quantity $C_{i} $(unit of currency) from the firm and at time $T_{N} =T$, the bond holder receives the face value $F$ and the last coupon $C_{N} $(unit of currency). 

4) The bond holder has the right to demand early redemption at any coupon dates $T_{i} ,(i=1, \cdots , N-1)$ prior to the maturity. If the bond holder demand early redemption, the firm does not pay the coupon of the day and bond holder receives the face value deducted the coupons he had already received. That is, if the bond holder demand early redemption at $T_{i} $, the bond holder receives $F-\sum _{j=1}^{i-1}C_{j}  $(unit of currency)

5) The default occurs only at time $T_{i} $ when the firm value is not sufficient to pay the debt and the coupon or the early redemptive money. If the default occurs, the bond holder receives $\delta \cdot V$ as default recovery. Here ${\rm 0}\le \delta < {\rm 1}$ is called a fractional \textit{default recovery}.

6) In the subinterval $(T_{i},T_{i+1} ]$, the prices of our corporate bond are given by sufficiently smooth functions $B_{i} (V,\; t),\; (i=0, \cdots , N-1)$, respectively.

\textbf{Mathematical model for bond price}

We will use the following notations for simplicity:
\[\bar{c}_{N} =F+C_{N} ;\; \; \bar{c}_{i} =C_{i} ,\; \; i=1,\cdots , N-1.\] 
First we analyse the default event. If the default event occurs or the bond holder demands early redemption before $T_{i} $ or at $T_{i} \; (i=1, \cdots , N-1)$, the bond contract doesn't exist on the interval $T_{i} <t\le T_{i+1} $. Hence $B_{i} (V,\; t){\rm ,}\; \; T_{i} <t\le T_{i+1} $ is the bond(debt) price on the interval $(T_{i} ,T_{i+1} ]$ under the \textit{condition} that the bond holder doesn't demand early redemption and the default event doesn't occur at $T_{i} $ or before $T_{i} $. Therefore, the fact that the default event doesn't occur at $T_{i} \; (i=1, \cdots , N-1)$ means that the firm value is not smaller than $B_{i} (V,\; T_{i} )$ after paying the coupon $\bar{c}_{i} $, that is,
$V\ge B_{i} (V,\; T_{i} )+\bar{c}_{i}$. 
And the fact that the default event doesn't occur at $T_{i} \; (i=1, \cdots , N-1)$ means that the firm value is enough to pay the early redemptive money to the bond holder at $T_{i} $. From the assumption 4), if the bond holder demands early redemtion at $T_{i} $, the firm gives the bond holder the face value deducted the coupons $\bar{c}_{j} \; (j=1, \cdots , i-1)$ that he had already received, namely, $F-\sum _{j=1}^{i-1}\bar{c}_{j}  $(unit of currency).
\noindent Thus the fact that the default event doesn't occur means that
$V\ge F-\sum _{j=1}^{i-1}\bar{c}_{j}$. 
On the whole, the fact that the default event doesn't occur at $T_{i} $ means that
\begin{equation} \label{eq1}
V\ge \max \left\{F-\sum_{j=1}^{i-1}\bar{c}_{j},\; \; B_{i} (V,\; T_{i} )+\bar{c}_{i} \right\}.
\end{equation}
On the other hand, the fact that the default event occurs at $T_{i} $ means that
\begin{equation} \label{eq2}
V<\max \left\{F-\sum _{j=1}^{i-1}\bar{c}_{j}  ,\; \; B_{i} (V,\; T_{i} )+\bar{c}_{i} \right\}.
\end{equation} 

Next, we analyse whether it is advantageous for the bond holder to keep the bond contract or to demand early redemption at $T_{i} $. If the bond holder demands early redemption at $T_{i} $, then the holder receives $F-\sum _{j=1}^{i-1}\bar{c}_{j}  $(unit of currency), whereas if the bond holder keeps the contract, then the holder receives coupon $\bar{c}_{i} $(unit of currency) and also possesses the bond with the value of $B_{i} (V,\; t)$ after $T_{i} $. Thus, at time $T_{i} $, the bond holder compares $B_{i} (V,\; T_{i} )+\bar{c}_{i} $ with $F-\sum _{j=1}^{i-1}\bar{c}_{j} $ and if $F-\sum _{j=1}^{i-1}\bar{c}_{j} $ is larger, then the holder will demand early redemption immediately but if $B_{i} (V,\; T_{i} )+\bar{c}_{i} $ is larger, then the holder will keep the contract. As a result, it is reasonable to think that the bond holder compares the proposal of keeping the contract with the proposal of demanding early redemption and then choose the better proposal. Therefore under the assumption that the default event didn't occur and the holder didn't demand early redemption at the coupon dates prior to $T_{i} $, if the default event doesn't occur at $T_{i} $, then the bond price at $T_{i} $ is
\[\max \left\{F-\sum _{j=1}^{i-1}\bar{c}_{j}  ,\; \; B_{i} (V,\; T_{i} )+\bar{c}_{i} \right\},\] 
and if the default event occurs at $T_{i} $, then the bond holder receives $\delta V$ as default recovery by the assumption 5). Hence the bond price at $T_{i}\;\left(i=1, \cdots ,N-1\right)$ is as follows:
\[\text{if}\;V\ge \max \left\{F-\sum _{j=1}^{i-1}\bar{c}_{j}  ,\; \; B_{i} (V,\; T_{i} )+\bar{c}_{i} \right\}, \text{then}\; \max \left\{F-\sum _{j=1}^{i-1}\bar{c}_{j}  ,\; \; B_{i} (V,\; T_{i} )+\bar{c}_{i} \right\},\]
\begin{equation}\label{eq3}
\text{if} \; V<\max \left\{F-\sum _{j=1}^{i-1}\bar{c}_{j}  ,\; \; B_{i} (V,\; T_{i} )+\bar{c}_{i} \right\}, \text{then}\; \delta V.
\end{equation}
\noindent In particular the bond price at the maturity $T_{N} =T$ is
\begin{equation} \label{eq4} 
B_{N-1} (V,\; T_{N} )=\bar{c}_{N} \cdot 1\{ V\ge \bar{c}_{N} \} +\delta V\cdot 1\{ V<\bar{c}_{N} \}.
\end{equation} 

From the assumptions 1), 2), 6),  it follows that the bond price $B_{i} $ on the subinterval $(T_{i} ,T_{i+1} )\;  (i={\rm 0},\cdots , N-1)$ satisfies the following PDE (this is derived in the standard way)\cite{Wil}:
\begin{equation} \label{eq5} 
\frac{\partial B_{i} }{\partial t} +\frac{s_{V}^{2} V^{2} }{2} \frac{\partial ^{2} B_{i} }{\partial V^{2} } +(r-b)V\frac{\partial B_{i} }{\partial V} -rB_{i} =0,\; \; T_{i} <t<T_{i+1} ,\; \; V>0. 
\end{equation}

From \eqref{eq1}-\eqref{eq4}, we get terminal conditions of the bond price.

\begin{align}
B_{N-1} (V,\; T_{N} )&=\bar{c}_{N} \cdot 1\{ V\ge \bar{c}_{N} \} +\delta V\cdot 1\{ V<\bar{c}_{N} \} ,\; \; \; V>0; \label{eq6} \\ 
B_{i} (V,\; T_{i+1} )&=\max \left\{B_{i+1} (V,\; T_{i+1} )+\bar{c}_{i+1},\; \; F-\sum _{j=1}^{i}\bar{c}_{j}  \right\}\cdot 1\left\{V\ge \max \left\{B_{i+1} (V,\; T_{i+1} )+\bar{c}_{i+1},\; \; F-\sum _{j=1}^{i}\bar{c}_{j}  \right\}\right\} \nonumber \\
&+\delta V\cdot 1\left\{V<\max \left\{B_{i+1} (V,\; T_{i+1} )+\bar{c}_{i+1} ,\; \; F-\sum _{j=1}^{i}\bar{c}_{j}  \right\}\right\}\; ,\; \; \; \; \; \; \; \; \; V>0,\; \; \; i=0,\cdots ,N-2. \label{eq7}
\end{align}

So our model of the bond price is \eqref{eq5}, \eqref{eq6}, \eqref{eq7}, that is, we must find $B_{i} $ satisfying \eqref{eq5}${}_{i=N-1} $ and \eqref{eq6} on the interval $T_{N-1} <t\le T_{N} $ and \eqref{eq5} and \eqref{eq7} on the interval $T_{i} <t\le T_{i+1} $ ($i=0,\cdots ,N-2$), respectively.

This problem on the interval $T_{N-1} <t\le T_{N} $ is just the same one as in \cite{OJKJ}. But on the interval $T_{i} <t\le T_{i+1} $ \; ($i=0,\cdots ,N-2$), the bond holder compares the proposal of keeping the bond with the proposal of demanding early redemption at $T_{i+1} $ and makes a decision, thus we must first find the early redemption boundary and the problem becomes American option-like pricing problem, or more precisely, Bermudan option-like pricing problem (at pages 193 of \cite{Jia} and 253$\mathrm{\sim}$255 of \cite{OK}).

\section{The existence and uniqueness of the default boundaries}
\noindent The following notations are used:
\begin{align}
f_{N-1}(V)&=\bar{c}_{N} \cdot 1\{ V\ge \bar{c}_{N} \} +\delta V\cdot 1\{ V<\bar{c}_{N} \},\label{eq8}\\
f_{i}(V)&=\max \left\{B_{i+1} (V,\; T_{i+1} )+\bar{c}_{i+1},\; \; F-\sum _{j=1}^{i}\bar{c}_{j}  \right\}\cdot 1\left\{V\ge \max \left\{B_{i+1} (V,\; T_{i+1} )+\bar{c}_{i+1},\; \; F-\sum _{j=1}^{i}\bar{c}_{j}  \right\}\right\} \nonumber \\
&+\delta V\cdot 1\left\{V<\max \left\{B_{i+1} (V,\; T_{i+1} )+\bar{c}_{i+1},\; \; F-\sum _{j=1}^{i}\bar{c}_{j}  \right\}\right\}\; ,\; \; \; \; V>0,\; \; \; i=0,1,\cdots ,N-2.\label{eq9}
\end{align}
The supremum and infimum of  the function $f$ defined on interval $[0,\; +\infty )$ are denoted by $M(f),\; m(f)$, repectively.

First consider the case when $i=N-1$. The bond price $B_{N-1} (V,t)$ on the interval $T_{N-1} <t\le T_{N} $ is the solution of the following problem:
\[\frac{\partial B_{N-1} }{\partial t} +\frac{s_{V}^{2} V^{2} }{2} \frac{\partial ^{2} B_{N-1} }{\partial V^{2} } +(r-b)V\frac{\partial B_{N-1} }{\partial V} -rB_{N-1} =0,\; \; T_{N-1} <t<T_{N} ,\; \; V>0,\] 
\[B_{N-1} (V,T_{N} )=\bar{c}_{N} \cdot 1\{ V\ge \bar{c}_{N} \} +\delta V\cdot 1\{ V<\bar{c}_{N} \} ,\; \; \; V>0.\] 
This is a terminal value problem for Black-Scholes equation with interest rate $r$, dividend rate $b$ and volatility $s_{V} $. By the terminal condition, $D_{N} =\bar{c}_{N} $ is the default boundary at $T_{N} $.

By the pricing formula of the first order binary option \cite{{Buc},{OK}} we have
\begin{equation} \label{eq10} 
B_{N-1} (V,\; t)=\bar{c}_{N} B_{\bar{c}_{N} }^{+} (V,\; t;\; T_{N} ;\; r,\; b,\; s_{V} )+\delta A_{\bar{c}_{N} }^{-} (V,\; t;\; T_{N} ;\; r,\; b,\; s_{V} ).         
\end{equation} 
Here $B_{K}^{+} (x, t; T; r, b, s_{V} ),\; \; A_{K}^{-} (x, t; T; r, b, s_{V} )$ is the price at $t$ of the bond and asset binary options with maturity $T$, exercise price $K$, interest rate $r$, dividend rate $b$ and volatility $s_{V} $, respectively \cite{OK}.

\begin{theorem}[The gradient estimates and the existence and uniqueness of the default boundaries] \label{theorem:1}
Assume that the volatility $s_{V} $ is enough large, that is, there exists a sequence $\delta =d_{N} <d_{N-1} <\; \cdots <d_{1} <1$ such that
\[s_{V} \ge \frac{(1-\delta )}{\sqrt{2\pi \cdot {\rm (}T_{i+1} -T_{i} {\rm )}} (d_{i} -d_{i+1} )}\;\;  \text{if} \;\;b=0;\]
\[s_{V} \ge \frac{(1-\delta )e^{-b(T_{i+1} -T_{i} )} }{\sqrt{2\pi \cdot {\rm (}T_{i+1} -T_{i} {\rm )}} (d_{i} -d_{i+1} e^{-b(T_{i+1} -T_{i} )} )}\;\; \text{if} \;\;b>0,i=1,\cdots ,N-1.\]
Then for the solution $B_{i} (V,t),\; i=1,\cdots ,N-1$ of \eqref{eq5}, \eqref{eq6}, \eqref{eq7}, we have
\begin{equation} \label{eq11} 
0<\partial _{V} B_{i} (V,\; T_{i} )<d_{i} <1 
\end{equation} 
and the equation
\[V=\max \left\{F-\sum _{j=1}^{i-1}\bar{c}_{j}  ,\; \; B_{i} (V,T_{i} )+\bar{c}_{i} \right\}\] 
has unique root $D_{i} $ and we have
\begin{equation} \label{eq12}
V\ge \max \left\{B_{i} (V,T_{i} )+\bar{c}_{i},\; \; F-\sum _{j=1}^{i-1}\bar{c}_{j}  \right\}\Leftrightarrow V\ge D_{i} .
\end{equation}
\end{theorem}

\begin{proof}
We use induction. First when $i=N-1$, we consider properties of $B_{N-1} (V,t)$. From \eqref{eq8}, the terminal payoff $f_{N-1} (V)$ is an discontinuous function with jump $\Delta f_{N-1} (\bar{c}_{N} )=(1-\delta )\bar{c}_{N} $ at $V=\bar{c}_{N} $. Now estimate $\partial _{V} B_{N-1} $. From Theorem 4 of \cite{OJK}, we have
\[m(f'_{N-1} )e^{-b(T_{N} -T_{N-1} )} +\frac{[\Delta f_{N-1} (\bar{c}_{N} )]^{-} }{\bar{c}_{N} } \cdot \frac{e^{-b(T_{N} -T_{N-1} )} }{\sqrt{2\pi \cdot s_{V}^{2} {\rm (}T_{N} -T_{N-1} {\rm )}} } <\partial _{V} B_{N-1} (V,T_{N-1} )<\]
\[<M(f'_{N-1} )e^{-b(T_{N} -T_{N-1} )} +\frac{[\Delta f_{N-1} (\bar{c}_{N} )]^{+} }{\bar{c}_{N} } \cdot \frac{e^{-b(T_{N} -T_{N-1} )} }{\sqrt{2\pi \cdot s_{V}^{2} {\rm (}T_{N} -T_{N-1} {\rm )}}}.\]
Here $[x]^{+} =\max \{ x,\; 0\} ,\; \; [x]^{-} =\min \{ x,\; 0\} $. From \eqref{eq8} $M(f'_{N-1} )=\delta =d_{N} ,\; \; m(f'_{N-1} )=0$ and thus we have
\[0<\partial _{V} B_{N-1} (V,T_{N-1} )<d_{N} e^{-b(T_{N} -T_{N-1} )} +\frac{(1-\delta )e^{-b(T_{N} -T_{N-1} )} }{\sqrt{2\pi \cdot s_{V}^{2} {\rm (}T_{N} -T_{N-1} {\rm )}} } .\] 
From the assumption of our theorem we have \eqref{eq11} for $i=N-1$.

Now consider roots of the non-linear equation
\[V=\max \left\{B_{N-1} (V,T_{N-1} )+\bar{c}_{N-1},\; \; F-\sum _{j=1}^{N-2}\bar{c}_{j}  \right\}.\] 
From \eqref{eq11} for $i=N-1$, the function
\[\max \left\{B_{N-1} (V,T_{N-1} )+\bar{c}_{N-1},\; \; F-\sum _{j=1}^{N-2}\bar{c}_{j}  \right\}\]
is monotone increasing on $V$ and its derivative is strictly less than 1 at all the potins except for the only indifferentiable point (the intersecting point of graphs of $B_{N-1} (V,T_{N-1} )+\bar{c}_{N-1}$ and $F-\sum _{j=1}^{N-2}\bar{c}_{j}  $). Thus the equation
\[V=\max \left\{B_{N-1} (V,T_{N-1} )+\bar{c}_{N-1},\; \; F-\sum _{j=1}^{N-2}\bar{c}_{j}  \right\}\] 
has unique root $D_{N-1} $. And from \eqref{eq11} for $i=N-1$ we have \eqref{eq12}.

Next we assume that we have \eqref{eq11} for $i=k+1$ and the equation
\[V=\max \left\{B_{k+1} (V,T_{k+1} )+\bar{c}_{k+1},\; \; F-\sum _{j=1}^{k}\bar{c}_{j}  \right\}\] 
has unique root $D_{k+1} $ and we have \eqref{eq12} for $i=k+1$. Then consider  properties of $B_{k} (V,\; t)$. From \eqref{eq9}, the terminal payoff $f_{k} (V)$ is an discontinuous function with jump $\Delta f_{k} (D_{k+1} )=(1-\delta )D_{k+1} $ at $V=D_{k+1} $. Now estimate $\partial _{V} B_{k} $.  From Theorem 4 of \cite{OJK}, we have
\[m(f'_{k} )e^{-b(T_{k+1} -T_{k} )} +\frac{[\Delta f_{k} (D_{k+1} )]^{-} }{D_{k+1} } \cdot \frac{e^{-b(T_{k+1} -T_{k} )} }{\sqrt{2\pi \cdot s_{V}^{2} {\rm (}T_{k+1} -T_{k} {\rm )}} } <\partial _{V} B_{k} (V,T_{k} )< \]
\[<M(f'_{k} )e^{-b(T_{k+1} -T_{k} )} +\frac{[\Delta f_{k} (D_{k+1} )]^{+} }{D_{k+1} } \cdot \frac{e^{-b(T_{k+1} -T_{k} )} }{\sqrt{2\pi \cdot s_{V}^{2} {\rm (}T_{k+1} -T_{k} {\rm )}} }.\] 
From \eqref{eq9} and \eqref{eq11} for $i=k+1$, we have $M(f'_{k} )=d_{k+1} ,\; \; m(f'_{k} )=0$ and thus we have
\[0<\partial _{V} B_{k} (V,T_{k} )<d_{k+1} e^{-b(T_{k+1} -T_{k} )} +\frac{(1-\delta )e^{-b(T_{k+1} -T_{k} )} }{\sqrt{2\pi \cdot s_{V}^{2} {\rm (}T_{k+1} -T_{k} {\rm )}} } .\] 
From the assumption of our theorem we have \eqref{eq11} for $i=k$. Now consider roots of the non-linear equation
\[V=\max \left\{B_{k} (V,T_{k} )+\bar{c}_{k},\; \; F-\sum _{j=1}^{k-1}\bar{c}_{j}  \right\}.\] 
From \eqref{eq11} for $i=k$, the function
\[\max \left\{B_{k} (V,T_{k} )+\bar{c}_{k},\; \; F-\sum _{j=1}^{k-1}\bar{c}_{j}  \right\}\] 
is monotone increasing on $V$ and its derivative is strictly less than 1 at all the potins except for the only indifferentiable point (the intersecting point of graphs of $B_{k} (V,T_{k} )+\bar{c}_{k}$ and $F-\sum _{j=1}^{k-1}\bar{c}_{j}  $). Thus the equation
\[V=\max \left\{B_{k} (V,T_{k} )+\bar{c}_{k},\; \; F-\sum _{j=1}^{k-1}\bar{c}_{j}  \right\}\] 
has unique root $D_{k} $. And from (11) for $i=k$ we have \eqref{eq12}.
\end{proof}

\begin{remark} \label{remark:1}
From \eqref{eq12}, $D_{i} $ is called \textit{the default boundary} at $T_{i} $.
\end{remark}

\begin{lemma}[The minimum estimate] \label{lemma:1}
Under the assumption of Theorem \ref{theorem:1}, for the solution $B_{i} (V,t),\; i=1,\cdots , N-1$ of \eqref{eq5}, \eqref{eq6}, \eqref{eq7}, we have the estimate:
\[\min _{V} B_{i} (V,T_{i} )=B_{i} (0,T_{i} )=0.\] 
\end{lemma}
The proof of Lemma \ref{lemma:1} is provided in appendix. In what follows, the proofs are provided in appendix if their mathematical proofs are not directly related to the expansion of the paper.

The following lemma shows that for our bond, if at one intermediate coupon date early redemption is always advantegeous \textit{regardless of firm value}, then early redemption is always advantegeous at all coupon dates prior to that.

\begin{lemma} \label{lemma:2}
Under the assumption of Theorem \ref{theorem:1}, if for some $i=2,\cdots , N-1$
\begin{equation} \label{eq13} 
\sup _{V} [B_{i} (V,\;T_{i} )+\bar{c}_{i} ]\le F-\sum _{j=1}^{i-1}\bar{c}_{j}  ,
\end{equation} 
then we have
\[\sup _{V} [B_{i-1} (V,\;T_{i-1} )+\bar{c}_{i-1} ]\le  F-\sum _{j=1}^{i-2}\bar{c}_{j}  .\] 
\end{lemma}

\begin{corollary} \label{corollary:1}
Under the assumption of Theorem \ref{theorem:1}, if
\begin{equation} \label{eq14} 
\sup _{V} [B_{1} (V,\; T_{1} )+\bar{c}_{1} ]>F,
\end{equation}
then we have
\[\sup _{V} [B_{i} (V,\; T_{i} )+\bar{c}_{i} ]>F-\sum _{j=1}^{i-1}\bar{c}_{j}  {\rm ,}\; \; i=1,2,\cdots ,N-1.\] 
\end{corollary}

From Lemma \ref{lemma:2}, if \eqref{eq13} holds for some $i\in \{ 2,\cdots ,N-1\} $, then \eqref{eq13} is also true for $i=1$, that is, $B_{1} (V,\; T_{1} )+C_{1} <F$ for all $V\in [0,+\infty )$. The financial meaning of this expression is that it is always advantageous for the bond holder to demand early redemption regardless of the firm value at the first coupon date ($T_{1} $). In the viewpoint of the bond issuing company, the significance of issuing bond is reduced in full width. Indeed, in this case the bond exists only on the interval $[0,T_{1} ]$ and does not exist after the time $T_{1} $. And the bond price on the interval $[0,T_{1} ]$ satisfies 
\[\frac{\partial B_{0} }{\partial t} +\frac{s_{V}^{2} V^{2} }{2} \frac{\partial ^{2} B_{0} }{\partial V^{2} } +(r-b)V\frac{\partial B_{0} }{\partial V} -rB_{0} =0,\; \; 0\le t<T_{1} ,\; \; V>0,\] 
\[B_{0} (V,\; T_{1} )=F\cdot 1\{ V\ge F\} +\delta V\cdot 1\{ V<F\} {\rm .}\] 
This problem is just the same as the problem (2.8) and (2.9) of \cite{OJKJ} and thus the bond price at $t\in [0,T_{1} ]$ is provided as
\[B_{0} (V,\; t)=FB_{F}^{+} (V,\; t;\; T_{1} ;\; r,\; b,\; s_{V} )+\delta A_{F}^{-} (V,\; t;\; T_{1} ;\; r,\; b,\; s_{V} ),\; \; 0\le t\le T_{1} .\] 
Therefore in this case the bond is a zero coupon bond with the maturity $T_{1} $ and the face value $F$.

It is reasonable for the coupon bond with early redemption to assume that \eqref{eq14} is satisfied.

\begin{lemma}[The maximum estimate] \label{lemma:3}
Under the assumption of Theorem 1 and the assumption \eqref{eq14}, the solution $B_{i} (V,\; t),\; i=1,\cdots ,N-1$ of \eqref{eq5},\eqref{eq6}  and \eqref{eq7} satisfies
\[\sup _{V} B_{i} (V,\; T_{i} )=B_{i} (+\infty ,\; T_{i} )=\sum _{j=i+1}^{N}\left[\bar{c}_{j} e^{-r(T_{j} -T_{i} )} \right] .\] 
\end{lemma}

Now using Lemma \ref{lemma:3} we analyse what relations between the parameters of the bond the assumption \eqref{eq14} requires. From Lemma \ref{lemma:3}, \eqref{eq14} becomes
\[\bar{c}_{1} +\sum _{j=2}^{N}\left[\bar{c}_{j} e^{-r(T_{j} -T_{1} )} \right] >F.\] 
Multiplying $e^{r(T_{N} -T_{1} )} $ to the both sides of this inequality, we have
\begin{equation} \label{eq15} 
\sum _{j=1}^{N}\left[C_{j} e^{r(T_{N} -T_{j} )} \right] >F\cdot e^{r(T_{N} -T_{1} )} -F.                     
\end{equation} 
This relation gives us the \textit{lower bound} for the bond coupons. In particular, if the intervals between the adjoining coupon dates and the coupons are always the same, i.e, $\Delta T=T_{i+1} -T_{i} ,\;  i=0,\cdots ,N-1$ and $C_{i} =C_{j} =C{\rm ,}\; 1\le i,\; j\le N$, then \eqref{eq15} becomes
\[C\sum _{j=1}^{N}\left[e^{r(T_{N} -T_{j} )} \right] >F\cdot [e^{r(T_{N} -T_{1} )} -1].\] 
This yields
\begin{equation} \label{eq16} 
\frac{C}{F} >(e^{r\Delta T} -1)\cdot \frac{e^{(N-1)r\Delta T} -1}{e^{Nr\Delta T} -1} \approx r\Delta T\cdot \frac{(N-1)r\Delta T}{Nr\Delta T} =\frac{N-1}{N} \cdot (r\Delta T).
\end{equation} 
This relation gives us the lower bound to the ratio of the coupon to the face value in the bond with early redemption provision.

\begin{remark} \label{remark:2}
From the process of deriving \eqref{eq15}, under the assumption of Theorem \ref{theorem:1}, the assumption \eqref{eq15} becomes a necessary condition for the assumption \eqref{eq14} to hold.
\end{remark}

Now in order to show that under the assumption of Theorem \ref{theorem:1}, the assumption \eqref{eq15} is a sufficient condition for the assumption \eqref{eq14} to hold, we prove the following lemmas.

\begin{lemma} \label{lemma:4}
If for some $1\le m\le N-1$,
\[\sum _{j=m}^{N}\left[\bar{c}_{j} e^{r(T_{N} -T_{j} )} \right] >\left(F-\sum _{j=1}^{m-1}\bar{c}_{j}  \right)\cdot e^{r(T_{N} -T_{m} )} \] 
is satisfied then we have
\[\sum _{j=m+1}^{N}\left[\bar{c}_{j} e^{r(T_{N} -T_{j} )} \right] >\left(F-\sum _{j=1}^{m}\bar{c}_{j}  \right)\cdot e^{r(T_{N} -T_{m+1} )} .\] 
\end{lemma}

The financial meaning of Lemma \ref{lemma:4} is that if a risk free discrete coupon bond has early redemption provision and keeping the risk free bond at some coupon date $T_{m} $ is more advantageous than early redemption, then keeping the risk free bond is  advantageous at the date $T_{m+1} $, too.

\begin{corollary} \label{corollary:2}
If \eqref{eq15} holds, then for $1\le m\le N-1$ we have
\[\sum _{j=m}^{N}\left[\bar{c}_{j} e^{r(T_{N} -T_{j} )} \right] >\left(F-\sum _{j=1}^{m-1}\bar{c}_{j}  \right)\cdot e^{r(T_{N} -T_{m} )} ,\] 
or equivalently,
\begin{equation} \label{eq17} 
\sum _{j=m}^{N}\left[\bar{c}_{j} e^{-r(T_{j} -T_{m-1} )} \right] >\left(F-\sum _{j=1}^{m-1}\bar{c}_{j}  \right)\cdot e^{-r(T_{m} -T_{m-1} )} .                
\end{equation} 
\end{corollary}

\begin{lemma} \label{lemma:5}
Under the assumption of Theorem \ref{theorem:1} and assumption \eqref{eq15}, we have
\[\sup _{V} B_{i} (V,\; T_{i} )=B_{i} (+\infty ,\; T_{i} )=\sum _{j=i+1}^{N}\left[\bar{c}_{j} e^{-r(T_{j} -T_{i} )} \right] .\] 
\end{lemma}

\begin{corollary} \label{corollary:3}
Under the assumption of Theorem \ref{theorem:1}, \eqref{eq15} is a necessary and sufficient condition for the assumption \eqref{eq14} to hold.
\end{corollary}

This shows that it is valid to set up the parameters of the discrete coupon bond with early redemption provision such that the assumption \eqref{eq15} is satisfied.

\begin{remark} \label{remark:3}
In the financial reality, there are such bonds that \eqref{eq15} is not satisfied. For example, see zero coupon bonds. In the case of such bonds that the assumption \eqref{eq15} is not satisfied (that is, the bonds with too small coupons), in the viewpoint of the bond issuing firm, the early redemption provision should be canceled. If the early redemption provision is canceled, then the bond becomes discrete coupon bond without early redemption provision (studied in \cite{OJKJ}, already).
\end{remark}
From now, we consider the bonds with early redemption provision satisfying \eqref{eq15}.

\section{The early redemption boundaries}

So far, it is not clear whether the problem \eqref{eq5} and \eqref{eq7} can be solved by using higher order binary options or not. To make it clear, we study the structure of the terminal functions $f_{i} (V)$ for $B_{i} (V{\rm ,}\; t),\; i=0,\cdots ,N-2$. 

First, we consider the positions of the graph of $y=B_{i+1} (V,\; T_{i+1} )+\bar{c}_{i+1} $ and the line of $y=F-\sum _{j=1}^{i}\bar{c}_{j}  $. See  Figure 1. There are 3 cases: the case that the inequality

\begin{equation} \label{eq18} 
\sup _{V} \{ B_{i+1} (V,\; T_{i+1} )+\bar{c}_{i+1} \} =\bar{c}_{i+1} +\sum _{j=i+2}^{N}\left[\bar{c}_{j} e^{-r(T_{j} -T_{i+1} )} \right] \le F-\sum _{j=1}^{i}\bar{c}_{j}
\end{equation}
holds(see line 1 in Figure 1), the case that the inequality

\begin{equation} \label{eq19} 
\min _{V} \{ B_{i+1} (V,\; T_{i+1} )+\bar{c}_{i+1} \} =\bar{c}_{i+1} >F-\sum _{j=1}^{i}\bar{c}_{j}   
\end{equation} 
holds(see line 2 in Figure 1) and the case that the graph of $y=B_{i+1} (V,\; T_{i+1} )+\bar{c}_{i+1} $ and the line of $y=F-\sum _{j=1}^{i}\bar{c}_{j}  $ intersect at only one point(see line 3 in Figure 1).

\begin{center}
\includegraphics*[width=14cm, height=10cm, keepaspectratio=false]{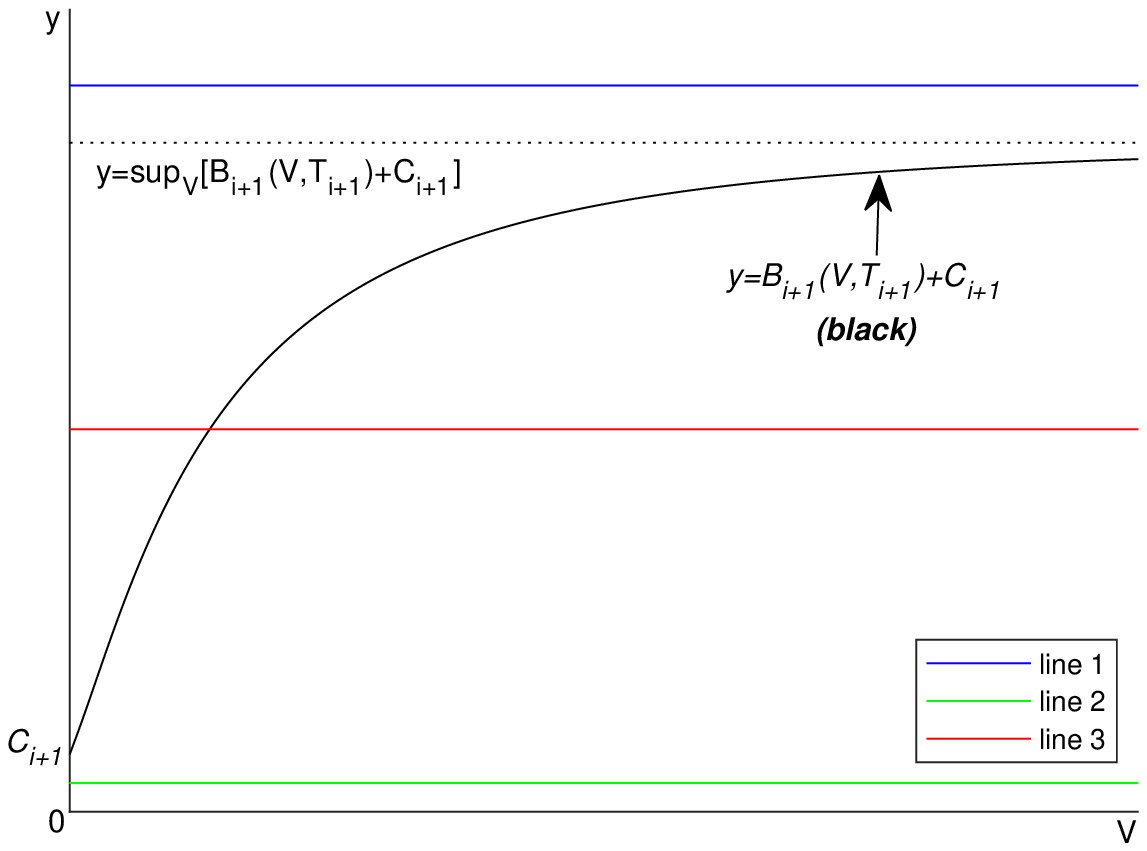}

\noindent Figure 1. The positions of the graph of $y=B_{i+1} (V,\; T_{i+1} )+\bar{c}_{i+1} $ and the line of $y=F-\sum _{j=1}^{i}\bar{c}_{j}  $. Here the lines 1, 2, 3 are the graphs of $y=F-\sum _{j=1}^{i}\bar{c}_{j}  $ in different cases.
\end{center}

Consider the first case (with \eqref{eq18}). Since we assumed \eqref{eq15}, Corollary \ref{corollary:3} implies that \eqref{eq14} holds and Corollary \ref{corollary:1} implies that
\[\sup _{V} [B_{i} (V,\; T_{i} )+\bar{c}_{i} ]>\; F-\sum _{j=1}^{i-1}\bar{c}_{j}  {\rm ,}\;\; i=1,2,\cdots ,N-1.\] 
This means that \eqref{eq18} can not hold for our bond. Therefore, we exclude the first case. Next, consider the second case (with \eqref{eq19}). \eqref{eq19} is equivalent to
\begin{equation} \label{eq20} 
\sum _{j=1}^{i+1}\bar{c}_{j}  >F 
\end{equation} 
and if for some $i=m$ \eqref{eq19} holds, then for all $m<i\le N-2$ \eqref{eq19} holds since $\bar{c}_{i} \ge 0$. That is, for all $m<i\le N-2$ and all $V\in [0,\; +\infty )$, we have
\[B_{i+1} (V,\; T_{i+1} )+\bar{c}_{i+1} >F-\sum _{j=1}^{i}\bar{c}_{j}  .\] 
Now let

\begin{equation} \label{eq21}
M={\mathop{\min }\limits_{0\le k\le N-1}} \left\{\sum _{j=1}^{k+1}\bar{c}_{j}  >F\right\}.                     
\end{equation} 
Then for $0\le i<M$, neither \eqref{eq18} nor \eqref{eq19} holds and we have only the third case (see line 3 in Figure 1).

\begin{remark} \label{remark:4}
The bond holder should keep always the contract at $T_{M+1} $ and the later coupon dates. Thus at $T_{M+1} $ and the later coupon dates, our discrete coupon bond becomes the discrete coupon bond without early redemption \cite{OJKJ}. Therefore, on the interval $T_{M+1} \le t\le T_{N} $, our bond price $B_{i} (V,\; t)\; \; (T_{i} \le t<T_{i+1} ,\; M\le i\le N-1)$ is given by the formula (2.10) in the case of $\lambda _{i} =0$ in \cite{OJKJ}. The result is as follows:
\begin{align}
B_{i} (V,\; t)&=\sum _{k=i}^{N-1}\Big[\bar{c}_{k+1} B_{D_{i+1} \cdots D_{k} D_{k+1} }^{\; +\; \; \; \cdots \; +\; \; +} (V,\; t;\; T_{i+1} ,\; \cdots ,\; T_{k} ,\; T_{k+1} ;\; r,\; b,\; s_{V} ).  \nonumber \\
&+\delta A_{D_{i+1} \cdots D_{k} D_{k+1} }^{\; +\; \; \; \cdots \; +\;  \; -} (V,\; t;\; T_{i+1} ,\; \cdots ,\; T_{k} ,\; T_{k+1} ;\; r,\; b,\; s_{V} )\Big]\; ,\; \; (T_{i} \le t<T_{i+1} ,\; \; M\le i\le N-1).\label{eq22}
\end{align}
\end{remark}
\noindent Thus, we have evaluated the bond price on the interval $T_{M} <t\le T_{N} $.

Now we only need to evaluate the bond price on the interval $T_{0} \le t\le T_{M} $.

\begin{theorem} [Existence and uniqueness of early redemption boundaries] \label{theorem:2}
Suppose that \eqref{eq15} is satisfied. Then for $i={\rm 1},\cdots ,M$ the nonlinear equation
\[B_{i} (V,\; T_{i} )+\bar{c}_{i} =F-\sum _{j=1}^{i-1}\bar{c}_{j}  \] 
has unique root $E_{i} $ and we have
\[B_{i} (V,\; T_{i} )+\bar{c}_{i} \ge F-\sum _{j=1}^{i-1}\bar{c}_{j}  \Leftrightarrow V\ge E_{i} .\] 
\end{theorem}
\begin{proof}
From \eqref{eq21}, for all $i={\rm 1},\cdots ,M$ we have
\[\bar{c}_{i} <F-\sum _{j=1}^{i-1}\bar{c}_{j}  .\] 
On the other hand, since \eqref{eq15} is satisfied, Corollay \ref{corollary:3} and Lemma \ref{lemma:2} implies
\[F-\sum _{j=1}^{i-1}\bar{c}_{j}  <\sup _{V} \left\{B_{i} (V,\; T_{i} )+\bar{c}_{i} \right\}.\] 
\noindent Now note that $\bar{c}_{i} =\min_{V} \left\{B_{i} (V,\; T_{i} )+\bar{c}_{i} \right\}$ we have
\[\min_{V} \left\{B_{i} (V,\; T_{i} )+\bar{c}_{i} \right\}<F-\sum _{j=1}^{i-1}\bar{c}_{j}  <\sup _{V} \left\{B_{i} (V,\; T_{i} )+\bar{c}_{i} \right\} \;\;\text{(see line 3 in Figure 1)}.\] 
The function $B_{i} (V,\; T_{i} )+\bar{c}_{i} $ is continuous and strictly increasing (see \eqref{eq11}), the nonlinear equation
\[B_{i} (V,\; T_{i} )+\bar{c}_{i} =F-\sum _{j=1}^{i-1}\bar{c}_{j}   ,\;\; \text{i.e.,}\;\;\; B_{i} (V,\; T_{i} )=F-\sum _{j=1}^{i}\bar{c}_{j}\]
\noindent has unique root $E_{i} $ and from \eqref{eq11} we have
\[B_{i} (V,\; T_{i} )+\bar{c}_{i} \ge F-\sum _{j=1}^{i-1}\bar{c}_{j}  \Leftrightarrow V\ge E_{i}.\]
\end{proof}

\begin{remark} \label{remark:5}
For all $i=1,\cdots ,M$ we have
\[B_{i} (V,\; T_{i} )+\bar{c}_{i} \ge F-\sum _{j=1}^{i-1}\bar{c}_{j}  \Leftrightarrow V\ge E_{i} .\] 
Thus $E_{i} $ is called the \textit{early redemption boundary} at $T_{i} $.
\end{remark}

\begin{remark} \label{remark:6}
If between the face value and coupons there is a relation
\[\sum _{j=1}^{N-1}C_{j}  <F,\]
then $M=N-1$ by \eqref{eq21}. Thus the early redemption boundary $E_{i} $ uniquely exists for any $i={\rm 1},\cdots ,N-1$ and our bond becomes the bond with early redemption in the whole interval $0\le t\le T_{N} $. On the other hand, if
\[\sum _{j=1}^{N-1}C_{j}  \ge F\] 
then from \eqref{eq21} we have $M<N-1$ and $E_{i} $ uniquely exists only for $i={\rm 1},\cdots ,M$. And our bond becomes the bond with early redemption on the interval $0\le t\le T_{M} $ but on the interval $T_{M} <t\le T_{N} $, it becomes the bond without early redemption and we calculate the bond price by \eqref{eq22}.
\end{remark}

\section{The pricing formula of our bond}

Now in order to get the formula of the bond price in the interval $0\le t\le T_{M} $, we calculate the bond price $B_{i} (V,\; T_{i+1} )$ at the coupon dates $T_{i} \; \; (i=0,\cdots , M-1)$.

In Figure 2 and Figure 3, the red curve is the graph of $y=\max \{ B_{i+1} (V,\; T_{i+1} )+$ $\bar{c}_{i+1} ,\; \; F-\sum _{j=1}^{i}\bar{c}_{j}  \} $; the blue curve is the graph of $y=B_{i+1} (V,\; T_{i+1} )+C_{i+1} $; the black line is the graph of $y=F-\sum _{j=1}^{i}\bar{c}_{j}  $; the pink line is the graph of $y=V$. As you can see in Figure 2 and Figure 3, there are two cases of the positions of $D_{i+1}$ and $E_{i+1}$: In first case the intersection point of the graph of  $y=V$ and the graph of $y=\max \big\{[B_{i+1} (V,\; T_{i+1})+\bar{c}_{i+1} ],\; F-\sum _{j=1}^{i}\bar{c}_{j}  \big\}$ is on the branch of $y=B_{i+1} (V,\; T_{i+1} )+\bar{c}_{i+1} $ so $D_{i+1}<E_{i+1}$. In second case the intersection point of the graph of $y=V$ and the graph of $y=\max \big\{[B_{i+1} (V,\; T_{i+1} )+\bar{c}_{i+1} ],\; \; F-\sum _{j=1}^{i}\bar{c}_{j}  \big\}$ is on the branch of $y=F-\sum _{j=1}^{i}\bar{c}_{j}  $ so $D_{i+1}>E_{i+1}$.

First we consider the case when the intersection point of the graph of $y=V$ and the graph of $y=\max \big\{[B_{i+1} (V,\; T_{i+1})+\bar{c}_{i+1} ],\; F-\sum _{j=1}^{i}\bar{c}_{j}  \big\}$ is on the branch of $y=B_{i+1} (V,\; T_{i+1} )+\bar{c}_{i+1} $. (Figure 2). Then $D_{i+1} $ is the solution of the equation $V=B_{i+1} (V,\; T_{i+1} )$ $+\bar{c}_{i+1} $ and $E_{i+1} \le D_{i+1} $. As you can see in Figure 2, in this case, the default boundary $D_{i+1} $ corresponds to the default event that occurs because of that the firm value is less than debt when the bond holder keeps the contract and we can rewrite the terminal payoff function at $T_{i+1} $ as follows:
\begin{equation} \label{eq23}
B_{i} (V,\; T_{i+1} )=[B_{i+1} (V,\; T_{i+1} )+\bar{c}_{i+1} ]\cdot 1\left\{V\ge D_{i+1} \right\}+\delta V\cdot 1\left\{V<D_{i+1} \right\}\; .
\end{equation}

\begin{center}
\includegraphics*[width=14cm, height=10cm, keepaspectratio=false]{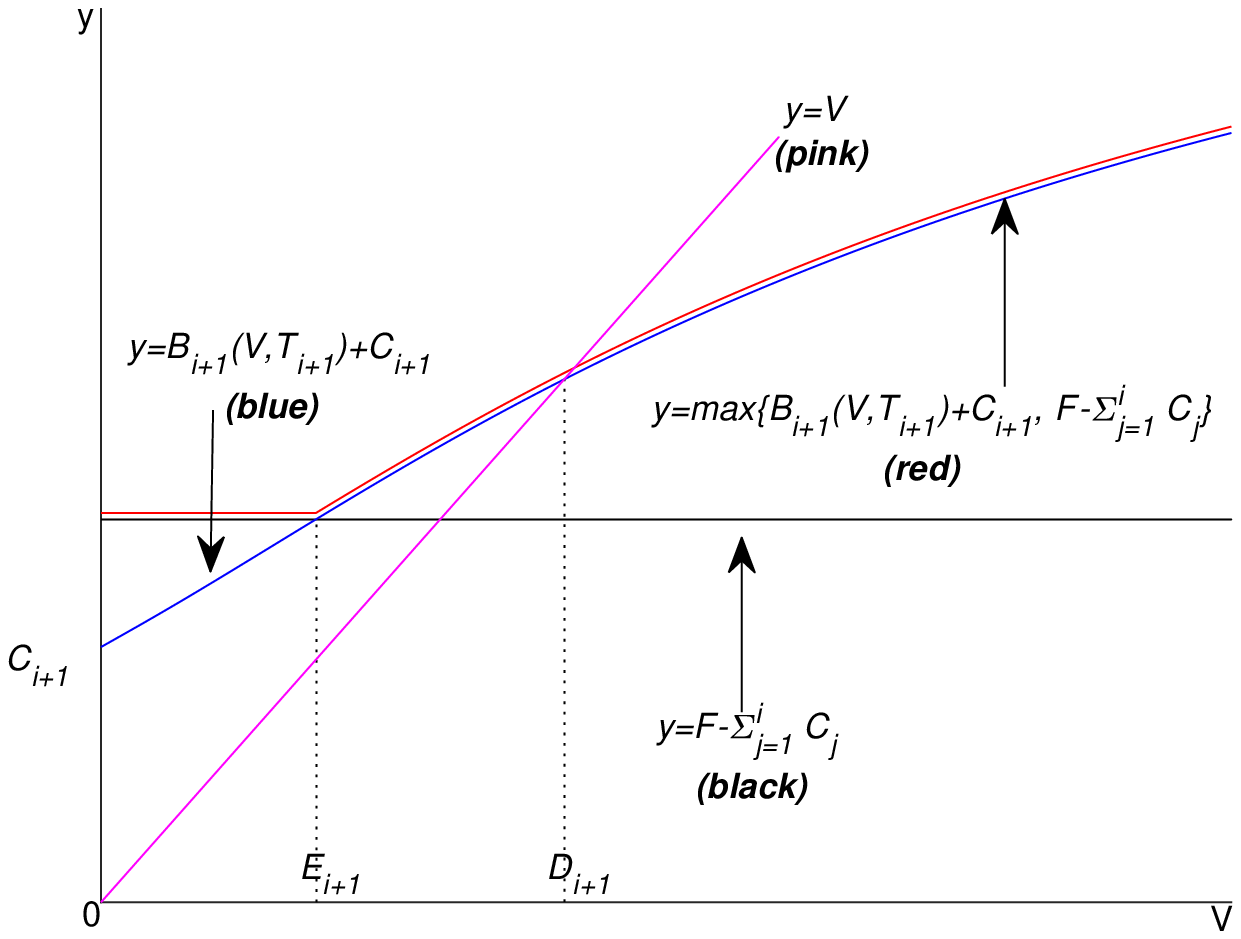}

Figure 2. The intersection point is on the branch of $y=B_{i+1} (V,\; T_{i+1} )+\bar{c}_{i+1} $.
\end{center}

\begin{center}
\includegraphics*[width=14cm, height=10cm,  keepaspectratio=false]{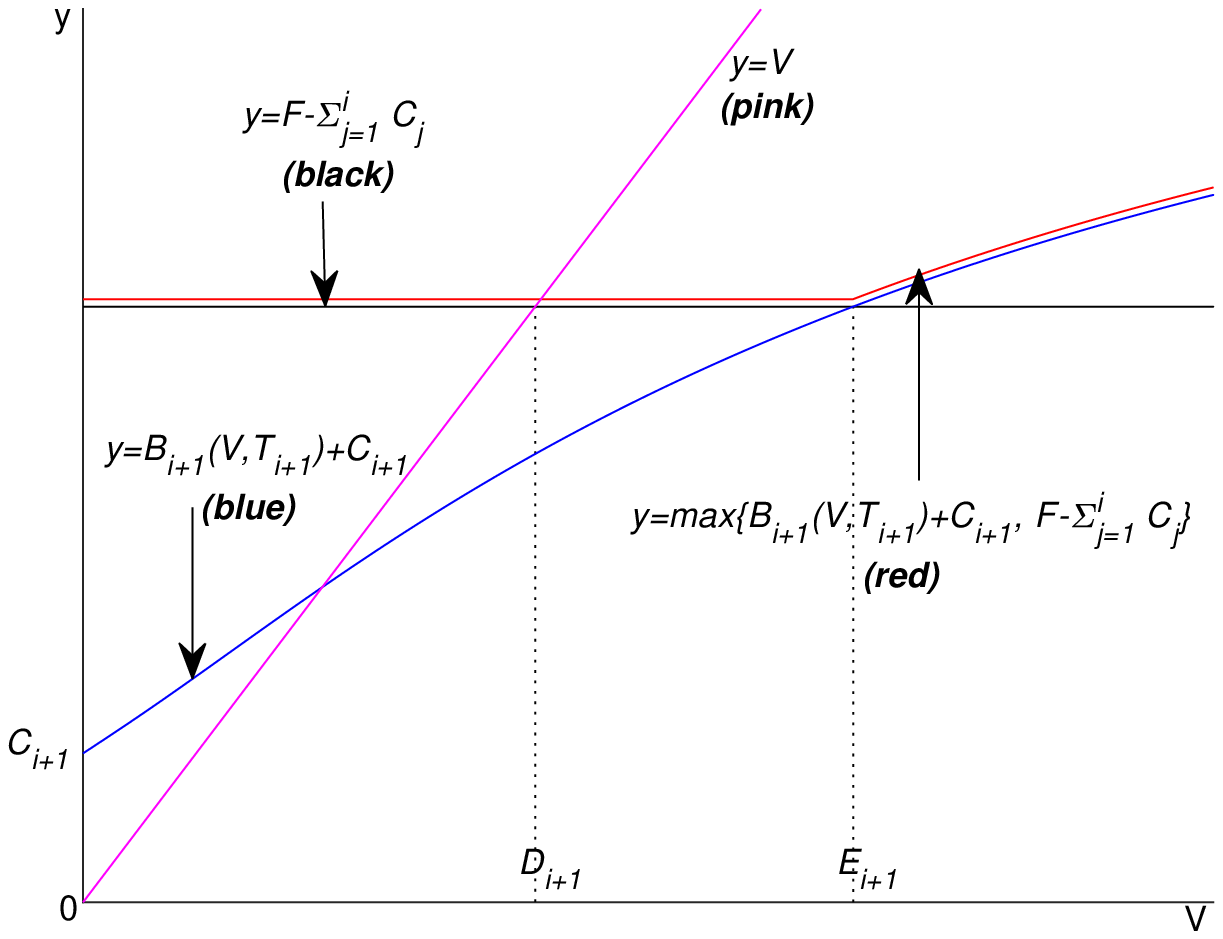}

\noindent Figure 3. The intersection point is on the branch of $y=F-\sum _{j=1}^{i}\bar{c}_{j}  $.
\end{center}

Next we consider the case when the intersection point of the graph of $y=V$ and the graph of $y=\max \big\{[B_{i+1} (V,\; T_{i+1} )+$\\
$\bar{c}_{i+1} ],\; \; F-\sum _{j=1}^{i}\bar{c}_{j}  \big\}$ is on the branch of $y=F-\sum _{j=1}^{i}\bar{c}_{j}  $ (Figure 3). Then $D_{i+1} =F-\sum _{j=1}^{i}\bar{c}_{j}  $ and $D_{i+1} <E_{i+1} $. As you can see in Figure 3, in this case, the default boundary $D_{i+1} $ corresponds to the default event that occurs because of that the firm value is less than the early redemption money when the bond holder demands early redemption and we can rewrite the terminal payoff function at $T_{i+1} $ as follows:

\begin{align}
B_{i} (V,\; T_{i+1} )&=[B_{i+1} (V,\; T_{i+1} )+\bar{c}_{i+1} ]\cdot 1\left\{V\ge E_{i+1} \right\}+(F-\Sigma_{j=1}^{i}\bar{c}_{j})\cdot 1\left\{D_{i+1} \le V<E_{i+1} \right\}+\delta V\cdot 1\left\{V<D_{i+1} \right\} \nonumber\\
&=[B_{i+1} (V,\; T_{i+1} )+\bar{c}_{i+1} ]\cdot 1\left\{V\ge E_{i+1} \right\}+(F-\Sigma _{j=1}^{i}\bar{c}_{j})\cdot \left(1\left\{V\ge D_{i+1} \right\}-1\left\{V\ge E_{i+1} \right\}\right)+\delta V\cdot 1\left\{V<D_{i+1} \right\}. \label{eq24}
\end{align}

\noindent Putting \eqref{eq23} and \eqref{eq24} together, then we have
\begin{align}
B_{i} (V,\; T_{i+1})&=[B_{i+1} (V,\; T_{i+1} )+\bar{c}_{i+1} ]\cdot 1\left\{V\ge \max\{ E_{i+1} ,\; D_{i+1} \} \right\}+ (F-\Sigma_{j=1}^{i}\bar{c}_{j})\cdot 1\{ D_{i+1} <E_{i+1}\}\nonumber\\
& \cdot \left[1\left\{V\ge \min \{ E_{i+1} ,\; D_{i+1} \} \right\}-1\left\{V\ge \max \{ E_{i+1} ,\; D_{i+1} \} \right\}\right]+\delta V\cdot 1\left\{V<D_{i+1} \right\}.\nonumber
\end{align}
Now we use the following notations:
\[U_{i+1} =\max \{ D_{i+1} ,\; E_{i+1} \} ,\;\; L_{i+1} =\min \{ D_{i+1} ,\; E_{i+1} \} , \;i=0,\cdots ,M-1.\]
Then we have
\begin{align}
B_{i} (V,\; T_{i+1} )&=[B_{i+1} (V,\; T_{i+1} )+\bar{c}_{i+1} ]\cdot 1\{ V\ge U_{i+1} \} +(F-\Sigma_{j=1}^{i}\bar{c}_{j})\cdot 1\{ D_{i+1} <E_{i+1} \} \cdot 1\{ V\ge L_{i+1} \}\nonumber\\
&-(F-\Sigma_{j=1}^{i}\bar{c}_{j})\cdot 1\{ D_{i+1} <E_{i+1} \} \cdot 1\{ V\ge U_{i+1} \} +\delta V\cdot 1\{ V<D_{i+1}\}.\label{eq25}
\end{align}

Thus for $i=M-1$ we have
\begin{align}
B_{M-1} (V,\; T_{M})&=[B_{M} (V,\; T_{M} )+\bar{c}_{M} ]\cdot 1\{ V\ge U_{M} \} +(F-\Sigma_{j=1}^{M-1}\bar{c}_{j})\cdot 1\{ D_{M} <E_{M} \} \cdot 1\{ V\ge L_{M} \} \nonumber \\
&-(F-\Sigma_{j=1}^{M-1}\bar{c}_{j})\cdot 1\{ D_{M} <E_{M} \} \cdot 1\{ V\ge U_{M} \} +\delta V\cdot 1\{ V<D_{M} \}. \nonumber
\end{align}
Substituting $B_{M} (V,\; T_{M} )$(\eqref{eq22} for $i=M-1$) to the above, we get
\begin{align}
B_{M-1} (V,\; T_{M})&=\sum _{k=M}^{N-1}\big[\bar{c}_{k+1} B_{D_{M+1} \cdots D_{k} D_{k+1} }^{\; +\; \; \cdots \;\;\;+\; \; +} (V,\; t;\; T_{M} ,\; \cdots ,\; T_{k+1} ;\; r,\; b,\; s_{V} ) \nonumber \\
&+\delta A_{D_{M+1} \; \cdots D_{k} D_{k+1} }^{ \; +\; \; \cdots \;\;\;\; +\; \; -} (V,\; t;\; T_{M} ,\; \cdots ,\; T_{k} ,\; T_{k+1} ;\; r,\; b,\; s_{V})\big]\cdot 1\{ V\ge U_{M} \} + \nonumber\\ 
&+\bar{c}_{M} \cdot 1\{ V\ge U_{M} \} +(F-\Sigma_{j=1}^{M-1}\bar{c}_{j})\cdot 1\{ D_{M} <E_{M} \} \cdot 1\{ V\ge L_{M} \} - \nonumber \\
&-(F-\Sigma_{j=1}^{M-1}\bar{c}_{j})\cdot 1\{ D_{M} <E_{M} \} \cdot 1\{ V\ge U_{M} \} +\delta V\cdot 1\{ V<D_{M}\} \nonumber
\end{align}
By using the pricing formula of higher order binary options\cite{OK}, we can get
\begin{align}
B_{M-1} (V,\; t)&=\sum _{k=M}^{N-1}\big[\bar{c}_{k+1} B_{U_{M} \; D_{M+1} \cdots D_{k} D_{k+1} }^{+\; \; \; \; +\; \; \; \; \cdots\; \; +\; \; +} (V,\; t;\; T_{M} ,\; \cdots ,\; T_{k+1} ;\; r,\; b,\; s_{V}) \nonumber \\
 &+\delta A_{U_{M} \; D_{M+1} \; \cdots D_{k} D_{k+1} }^{+\; \; \; \; +\; \; \; \; \; \cdots \;\; +\; \; -} (V,\; t;\; T_{M} ,\; T_{M+1} ,\; \cdots ,\; T_{k+1} ;\; r,\; b,\; s_{V} )\big] \nonumber\\
&+(F-\Sigma_{j=1}^{M-1}\bar{c}_{j})\cdot 1\{ D_{M} <E_{M} \} \cdot B_{L_{M} }^{+} (V,\; t;\; T_{M} ;\; r,\; b,\; s_{V} )- \label{eq26} \\
&-(F-\Sigma_{j=1}^{M-1}\bar{c}_{j})\cdot 1\{ D_{M} <E_{M} \} \cdot B_{U_{M} }^{+} \cdot (V,\; t;\; T_{M} ;\; r,\; b,\; s_{V} )+\nonumber \\
&+\bar{c}_{M} B_{U_{M} }^{+} (V,\; t;\; T_{M} ;\; r,\; b,\; s_{V} )+\delta A_{D_{M} }^{-} \cdot (V,\; t;\; T_{M} ;\; r,\; b,\; s_{V} ).\nonumber 
\end{align}
We can rewrite this as follows:
\begin{align}
B_{M-1} (V,\; t)&=\sum_{k=M-1}^{N-1}\big[\bar{c}_{k+1} B_{U_{M} \; D_{M+1} \cdots D_{k} D_{k+1} }^{+\; \; \; \; +\; \; \; \; \cdots \;\; +\; \; +} (V,\; t;\; T_{M} ,\; \cdots ,\; T_{k+1} ;\; r,\; b,\; s_{V} ) \nonumber \\
&+\delta A_{U_{M} \; D_{M+1} \; \cdots D_{k} D_{k+1} }^{+\; \; \; \; +\; \; \; \; \; \cdots \;\; +\; \; -} (V,\; t;\; T_{M} ,\; T_{M+1} ,\; \cdots ,\; T_{k+1} ;\; r,\; b,\; s_{V} )\big] \nonumber \\
&+(F-\Sigma_{j=1}^{M-1}\bar{c}_{j})\cdot 1\{ D_{M} <E_{M} \} \cdot B_{L_{M} }^{+} (V,\; t;\; T_{M} ;\; r,\; b,\; s_{V} )- \nonumber \\
&-(F-\Sigma_{j=1}^{M-1}\bar{c}_{j})\cdot 1\{ D_{M} <E_{M} \} \cdot B_{U_{M} }^{+} \cdot (V,\; t;\; T_{M} ;\; r,\; b,\; s_{V} ). \nonumber
\end{align} 
By induction, we assert that the following theorem holds.

For convenience, we use the following notations:
\begin{equation} \label{eq27} 
D_{N} =E_{N} =\bar{c}_{N} ;\;\; U_{i} =D_{i} ,\; \; M<i\le N.
\end{equation} 
\begin{theorem}[The pricing formula] \label{theorem:3}
Under the assumptions of Theorem \ref{theorem:1} and Theorem \ref{theorem:2}, the solution of \eqref{eq5} and \eqref{eq6} and \eqref{eq7} is given as follows: For $i=0,\; \cdots ,\; N-1$,
\begin{align}
B_{i} (V,\; t)&=\sum_{k=i}^{N-1}\big[\bar{c}_{k+1} B_{U_{i+1} \cdots U_{k} U_{k+1} }^{\; +\; \cdots \;\; +\; \; \; +} (V,\; t;\; T_{i+1} ,\; \cdots ,\; T_{k+1} ;\; r,\; b,\; s_{V} ) \nonumber\\
&+\delta A_{U_{i+1} \; \cdots U_{k} D_{k+1} }^{\; +\; \; \cdots \;\; +\; \; -} (V,\; t;\; T_{i+1} ,\; \cdots ,\; T_{k} ,\; T_{k+1} ;\; r,\; b,\; s_{V} )\big]+ \nonumber\\
&+\sum_{k=i+1}^{M}(F-\Sigma_{j=1}^{k-1}\bar{c}_{j})\cdot  1\{D_{k} <E_{k}\} \cdot \big[B_{U_{i+1} \cdots U_{k-1} L_{k} }^{\; +\; \; \cdots \; \; +\; \; \; \; +} (V,\; t;\; T_{i+1} ,\; \cdots ,\; T_{k} ;\; r,\; b,\; s_{V} ) \label{eq28} \\
&-B_{U_{i+1} \cdots U_{k-1} U_{k} }^{\; +\; \; \cdots \; \; +\; \; \; \; +} (V,\; t;\; T_{i+1} ,\; \cdots ,\; T_{k} ;\; r,\; b,\; s_{V} )\big]\; ,\; \; \; T_{i} <t<T_{i+1}. \nonumber
\end{align}
Here $B_{K_{1} \cdots K_{m} }^{+\; \cdots \; +} (x,\; t;\; T_{1} ,\; \cdots ,\; T_{m} ;\; r,\; q,\; \sigma )$ and $A_{K_{1} \cdots K_{m} }^{+\; \cdots \; +} (x,\; t;\; T_{1} ,\; \cdots ,\; T_{m} ;\; r,\; q,\; \sigma )$ is the price of $m$-order bond and asset binary options \cite{OK} with the free risk rate $r$, the dividend rate $q$ and the volatility $\sigma $, respectively.
\end{theorem}

\begin{corollary} \label{corollary:4}
Under the assumptions of Theorem 1 and Theorem 2, the solution of \eqref{eq5},\eqref{eq6}, \eqref{eq7} is given as follows: For $i=0,\cdots ,N-1$,
\begin{align}
B_{i} (V,\; t)&=\sum _{k=i}^{N-1}\big[e^{-r(T_{k+1} -t)} \bar{c}_{k+1} N_{k-i+1} (d_{i+1}^{-} (t),\; \cdots ,\; d_{k+1}^{-} (t);\; A_{k-i+1} ) \nonumber \\
&+e^{-b(T_{k+1} -t)} \delta VN_{k-i+1} (d_{i+1}^{+} (t),\; \cdots ,d_{k}^{+} (t),\; -\tilde{d}_{k+1}^{+} (t);\; A_{k-i+1}^{-} )\big]\; + \label{eq29} \\
&+\sum _{k=i+1}^{M}(F-\Sigma_{j=1}^{k-1}\bar{c}_{j})\cdot  1\{ D_{k} <E_{k} \} \cdot e^{-r(T_{k} -t)} \big[N_{k-i} (d_{i+1}^{-} (t),\; \cdots ,\; d_{k-1}^{-} (t),\; \tilde{\tilde{d}}_{k}^{-} (t);\; A_{k-i} ) \nonumber \\
&-N_{k-i} (d_{i+1}^{-} (t),\; \cdots ,\; d_{k-1}^{-} (t),\; d_{k}^{-} (t);\; A_{k-i} )\big]\; ,\; \; \; \; T_{i} <t<T_{i+1} . \nonumber
\end{align} 
Here
\[d_{j}^{\pm } (t)=\frac{1}{s_{V} \sqrt{T_{j} -t} } \left[\ln \frac{V}{U_{j} } +\left(r-b\pm \frac{1}{2} s_{V} {}^{2} \right)(T_{j} -t)\right],\] 
\[\tilde{d}_{j}^{\pm } (t)=\frac{1}{s_{V} \sqrt{T_{j} -t} } \left[\ln \frac{V}{D_{j} } +\left(r-b\pm \frac{1}{2} s_{V} {}^{2} \right)(T_{j} -t)\right],\] 
\[\tilde{\tilde{d}}_{j}^{\pm } (t)=\frac{1}{s_{V} \sqrt{T_{j} -t} } \left[\ln \frac{V}{L_{j} } +\left(r-b\pm \frac{1}{2} s_{V} {}^{2} \right)(T_{j} -t)\right].\]
And the matrix $A_{n} ,\; \; A_{n}^{-} $ are the inverse $A_{n} =\left(R_{n} \right)^{-1} ,\; \; A_{n}^{-} =\left(R_{n}^{-} \right)^{-1} $of the $n\times n$ dimensional matrix $R_{n} ,\; \; R_{n}^{-} $, which are given as follows: if we use the notation $\tilde{T}_{j} =T_{j+i+1} , \; 0\le j<n$ then the elements $r_{lm} (t)$ of $R_{n} $ are given as
\[r_{ll} (t)=1,\; \; r_{lm} (t)=r_{ml} (t)=\sqrt{\frac{\tilde{T}_{l} -t}{\tilde{T}_{m} -t} } ,\; \; l<m,\; \; (l,\; m=0,\cdots ,n-1).\] 
And $R_{n}^{-} $ is the matrix, in which,
\[r_{m,\; n-1}^{-} (t)=-r_{m,\; n-1} (t), r_{n-1,\; m}^{-} (t)=-r_{n-1,\; m} (t), (m=0,\cdots ,n-2)\] 
and the other elements are as in $R_{n} $. And $N_{n} (d_{1} ,\; \cdots ,\; d_{n} ;\; A_{n} )$ is $n$-dimensional normal distribution function with correlation matrix $R_{n} =\left(A_{n} \right)^{-1} $, i.e.,
\[N_{n} (d_{1} ,\; \cdots ,\, d_{n} ;\; A_{n} )=\int _{-\infty }^{d_{1} }\cdots \int _{-\infty }^{d_{n} }\frac{1}{(\sqrt{2\pi } )^{n} } \sqrt{\det A_{n} } \exp \left(-\frac{1}{2} y^{\bot } A_{n} y\right)dy  .\]
\end{corollary}

\begin{corollary}[The initial bond price] \label{corollary:5}
Under the assumptions of Theorem \ref{theorem:1} and Theorem \ref{theorem:2}, the initial bond price is denoted as follows:
\begin{align}
B_{0}=B_{0} (V_{0} ,\; 0)&=\sum _{k=0}^{N-1}\big[e^{-rT_{k+1} } \bar{c}_{k+1} N_{k+1} (d_{1}^{-} (0),\; \cdots ,\; d_{k+1}^{-} (0);\; A_{k+1} ) \nonumber \\
&+e^{-bT_{k+1} } \delta V_{0} N_{k+1} (d_{1}^{+} (0),\; \cdots ,d_{k}^{+} (0),\; -\tilde{d}_{k+1}^{+} (0);\; A_{k+1}^{-} )\big]+ \label{eq30}\\
&+\sum _{k=1}^{M}(F-\Sigma_{j=1}^{k-1}\bar{c}_{j})\cdot  1\{ D_{k} <E_{k} \} \cdot e^{-rT_{k} } \big[N_{k} (d_{1}^{-} (0),\; \cdots ,\; d_{k-1}^{-} (0),\; \tilde{\tilde{d}}_{k}^{-} (0);\; A_{k})- \nonumber \\
&-N_{k} (d_{1}^{-} (0),\; \cdots ,\; d_{k-1}^{-} (0),\; d_{k}^{-} (0);\; A_{k} )\big] \nonumber
\end{align}
Now we denote the initial leverage ratio $F/V_{0} $ by $L$ and the $k$th coupon ratio $C_{k} /F$ by $c_{k} $. Then the initial price of the bond is as follows:
\begin{align}
&B_{0} (V_{0} ,\; 0)=B_{0} (L,F,c_{1} ,\cdots ,c_{N} ;\delta ;r,b)=F\Big\{e^{-rT_{N}} N_{N} (d_{1}^{-} (0),\; \cdots ,\; d_{N}^{-} (0);\; A_{N} )+ \nonumber \\
&+\sum _{k=0}^{N-1}\Big[e^{-rT_{k+1} {}^{{}^{} } } c_{k+1} N_{k+1} (d_{1}^{-} (0),\; \cdots ,\; d_{k+1}^{-} (0);\; A_{k+1} )+e^{-bT_{k+1} } \frac{\delta }{L} N_{k+1} (d_{1}^{+} (0),\; \cdots ,d_{k}^{+} (0),\; -\tilde{d}_{k+1}^{+} (0);\; A_{k+1}^{-} )\Big]+\label{eq31}\\
&+\sum_{k=1}^{M}(1-\Sigma_{j=1}^{k-1}c_{j}  )\cdot  1\{ D_{k} <E_{k} \} \cdot e^{-rT_{k} } \big[N_{k} (d_{1}^{-} (0),\; \cdots ,\; d_{k-1}^{-} (0),\; \tilde{\tilde{d}}_{k}^{-} (0);\; A_{k} )-N_{k} (d_{1}^{-} (0),\; \cdots ,\; d_{k-1}^{-} (0),\; d_{k}^{-} (0);\; A_{k})\big]\Big\} \nonumber
\end{align}
\end{corollary}

In what follows, we give numerical examples for the bond price calculated using the pricing formula \eqref{eq29} and Matlab. We use the function \textit{mvncdf} of Matlab in order to calculate multi-dimensional normal distribution function in the pricing formula \eqref{eq29}. The basic data are given as follows:
\[N=3,\; \; T_{1} =1,\; \; T_{2} =2,\; \; T_{3} =3(annum),\]
\[ r=0.03,\; \; b=0,\; \; s_{V} =1.0,\; \; \delta =0.5,\; \; F=1000,\; \; C_{1} =C_{2} =C_{3} =40.\]

\vspace{5mm}
Figure 4 shows the default boundary $D_{1} $ and the early redemption boundary $E_{1} $ at $T_{1}$ when the firm value $V$ varies from 0 to 16000. Here the red line represents the graph-$(V,\; \max \{ B_{1} (V,\; T_{1} )+C_{1} ,\; F\} )$, the blue curve represents the $(V,\; B_{1} (V,\; T_{1} )+C_{1} )$-graph and the black line represents $(V,\; F)$- graph and the pink line represents $(V,\; V)$-graph, respectively.

Figure 5 shows the default boundary $D_{2} $ and the early redemption boundary $E_{2} $ at $T_{2}$ when the firm value $V$ varies from 0 to 16000. Here the red line represents the graph-$(V,\; \max \{ B_{2} (V,\; T_{2} )+C_{2} ,\; F-C_{1} \} )$, the blue curve represents the $(V,\; B_{2} (V,\; T_{2} )+C_{2} )$-graph and the black line represents the $(V,\; F-C_{1} )$-graph and the pink line represents the $(V,\; V)$-graph, respectively.

\vspace{1mm}
As you can see in Figure 4 and Figure 5, for $i=1,2$ we have:

If $0\le V<D_{i} $, then the firm value is in the default region. That is, the default event occurs at $T_{i} $;

If $D_{i} \le V<E_{i} $, then the firm value is in the early redemption region. That is, the bond holder should demand early redemption at $T_{i} $;

If $E_{i} \le V$, then the firm value is in the continuous region. That is, the bond holder should keep the contract at $T_{i} $.

\begin{center}
\includegraphics*[width=14cm, height=10cm, keepaspectratio=false]{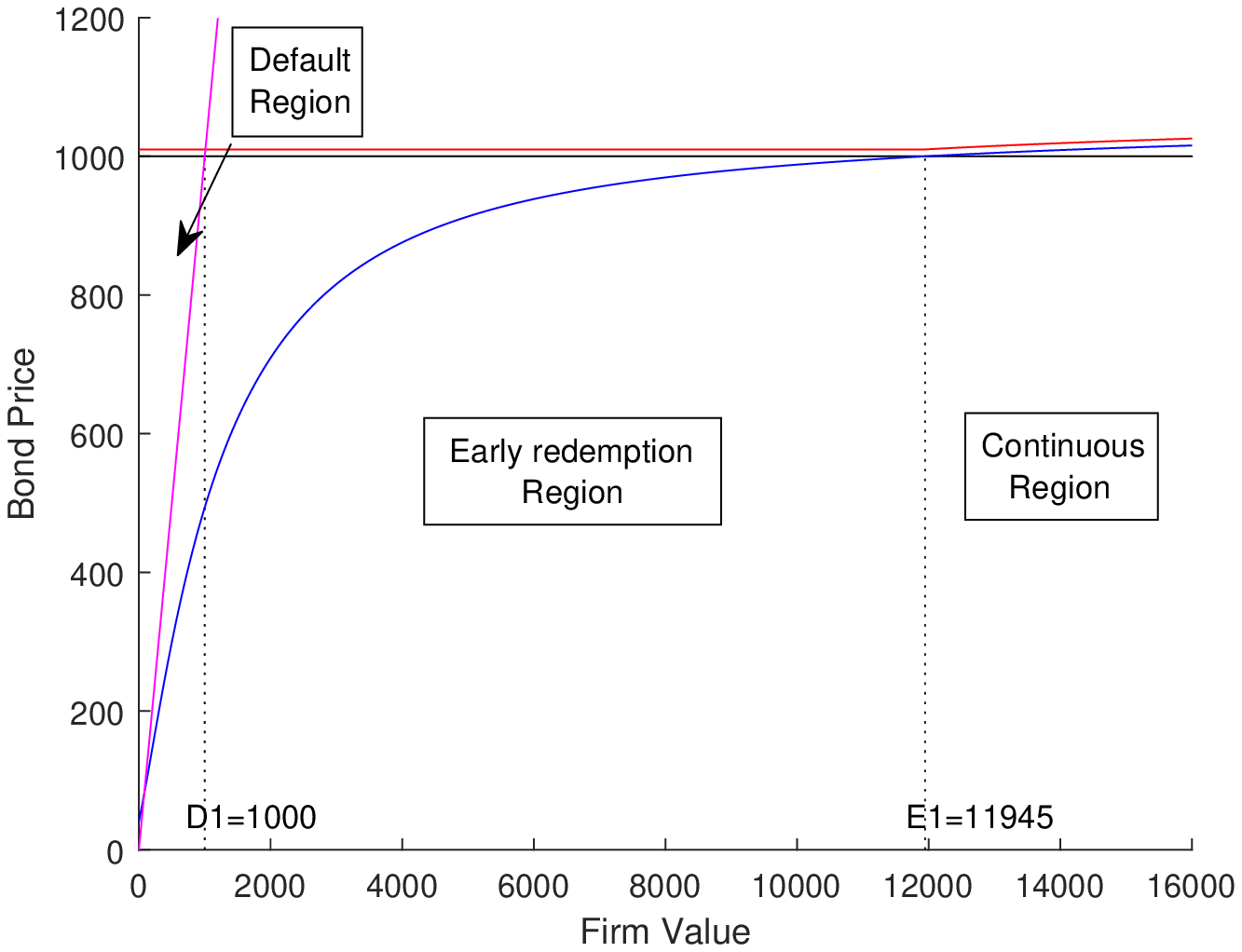}

Figure 4. The default boundary $D_{1} $ and the early redemption boundary $E_{1} $ at $T_{1} $.
\end{center}

\begin{center}
\includegraphics*[width=14cm, height=10cm, keepaspectratio=false]{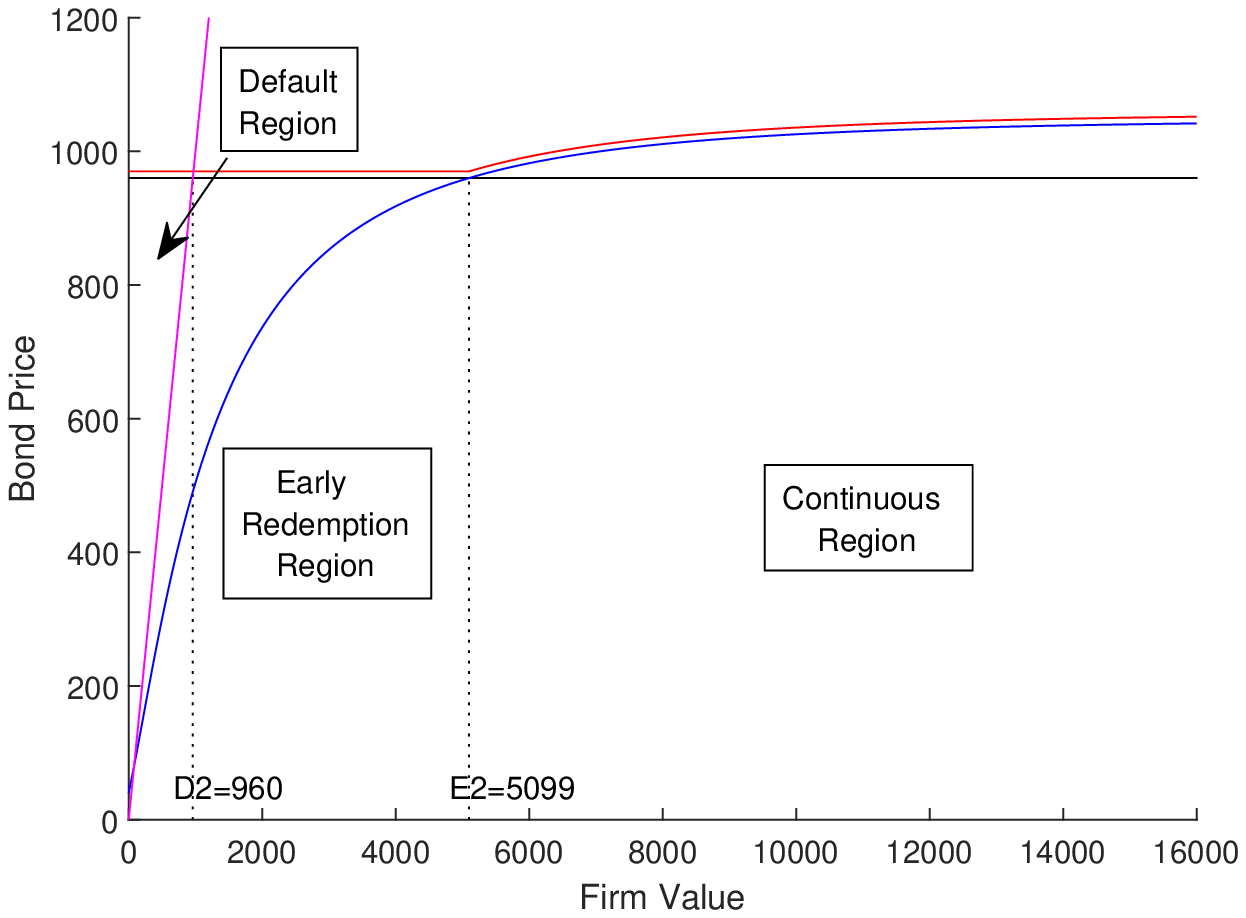}

Figure 5. The default boundary $D_{2} $ and the early redemption boundary $E_{2} $ at $T_{2} $.
\end{center}

\begin{center}
\includegraphics*[width=14cm, height=10cm, keepaspectratio=false]{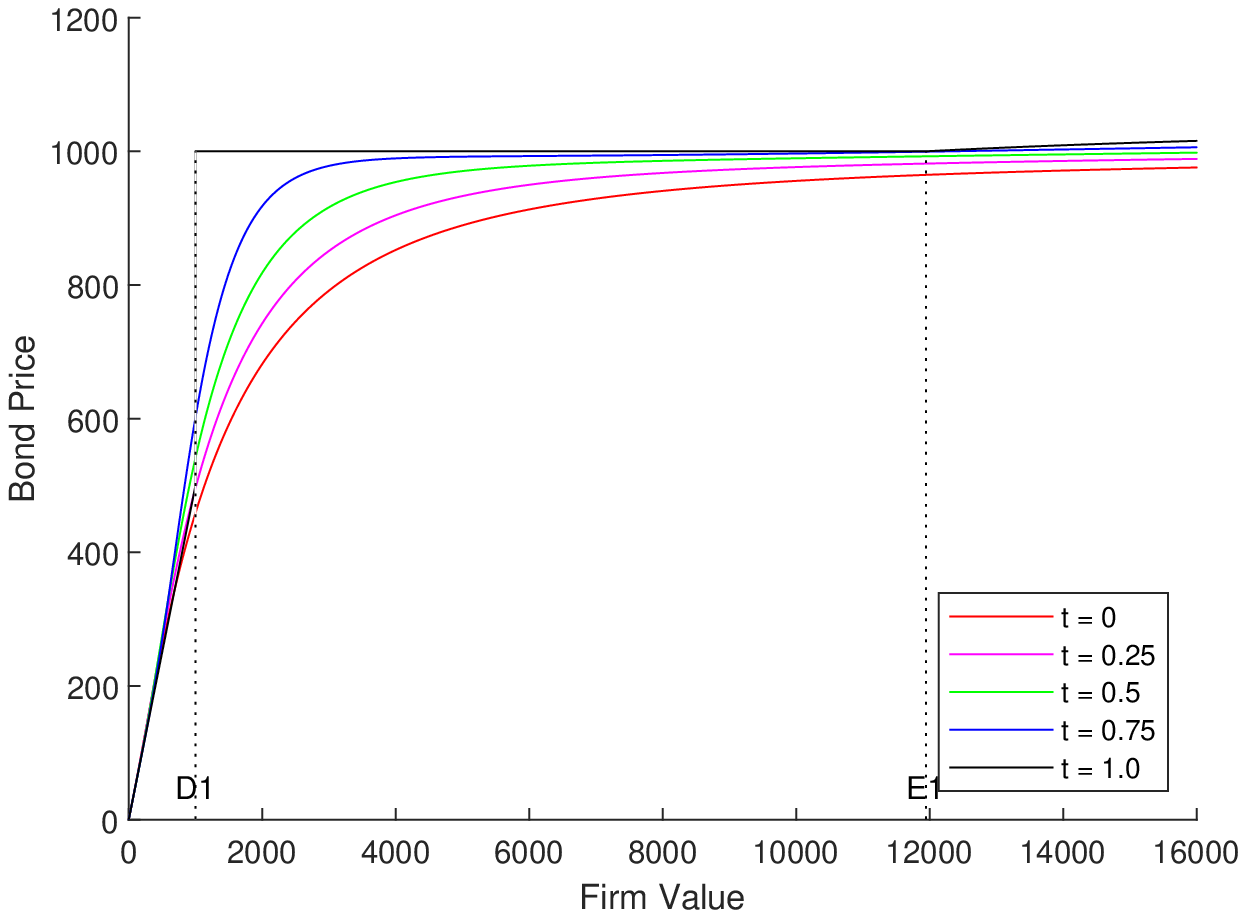}

Figure 6. $(V,\; B_{0} (V,\; t))$-graphs
\end{center}

\begin{center}
\includegraphics*[width=14cm, height=10cm, keepaspectratio=false]{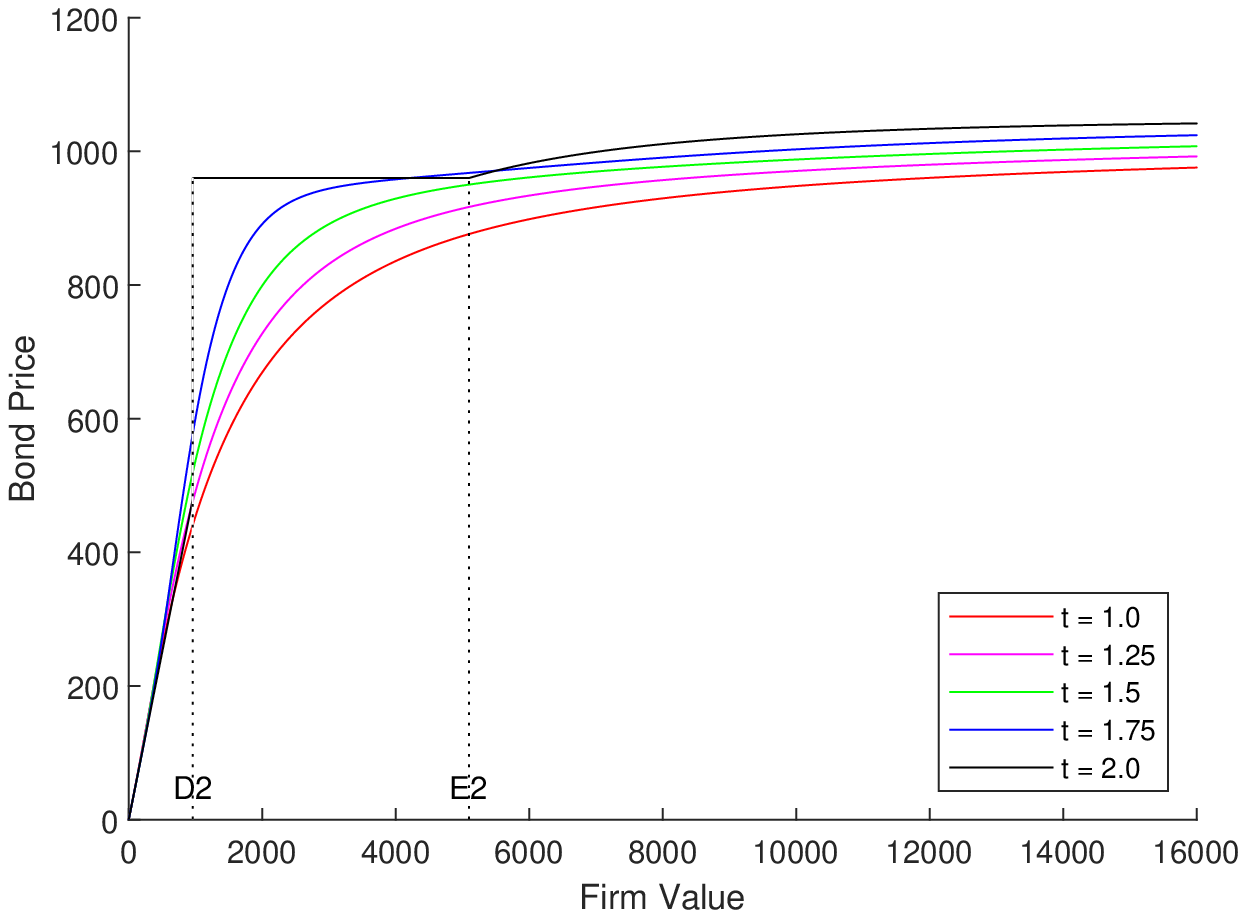}

Figure 7. $(V,\; B_{1} (V,\; t))$-graphs
\end{center}

\begin{center}
\includegraphics*[width=14cm, height=10cm, keepaspectratio=false]{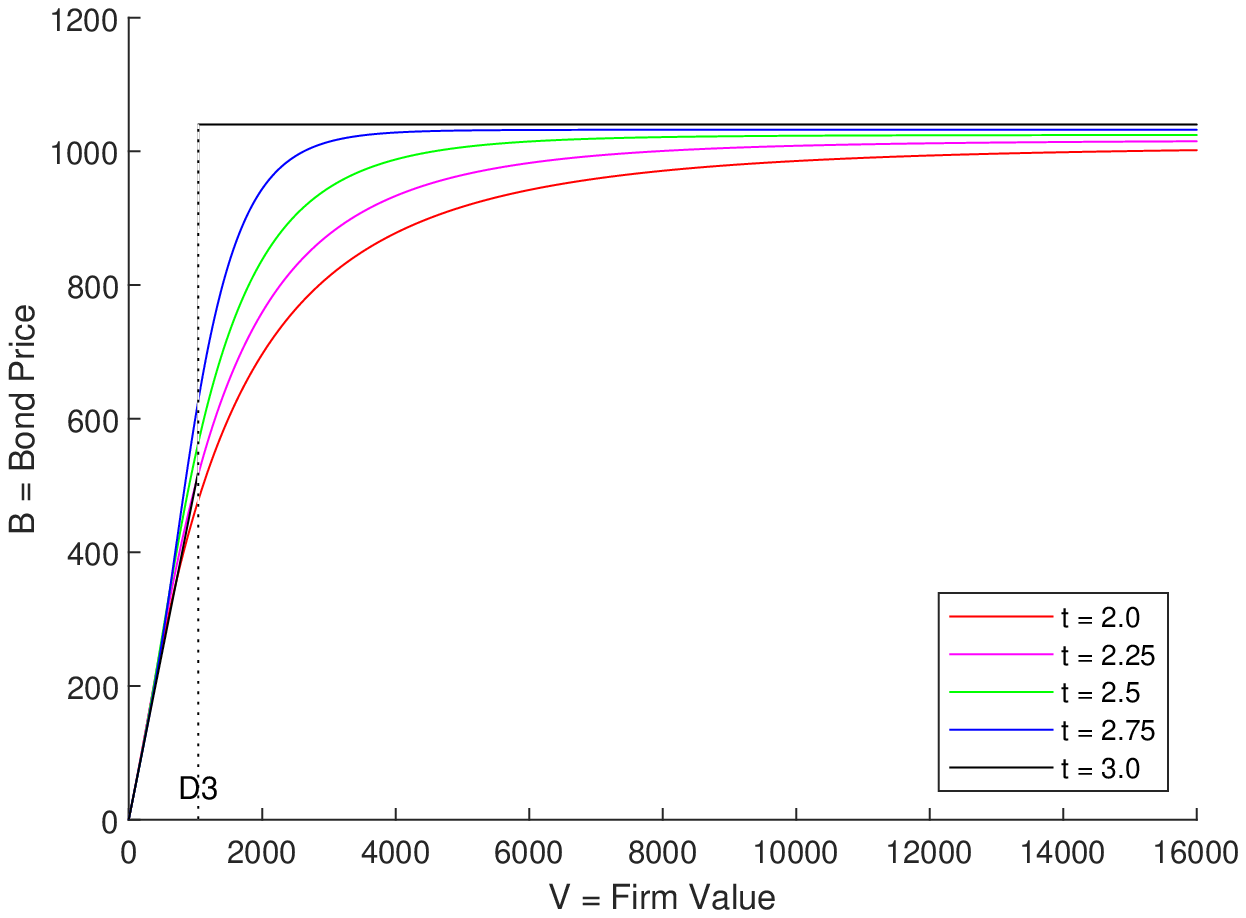}

Figure 8. $(V,\; B_{2} (V,\; t))$-graphs
\end{center}

\begin{center}
\includegraphics*[width=14cm, height=10cm, keepaspectratio=false]{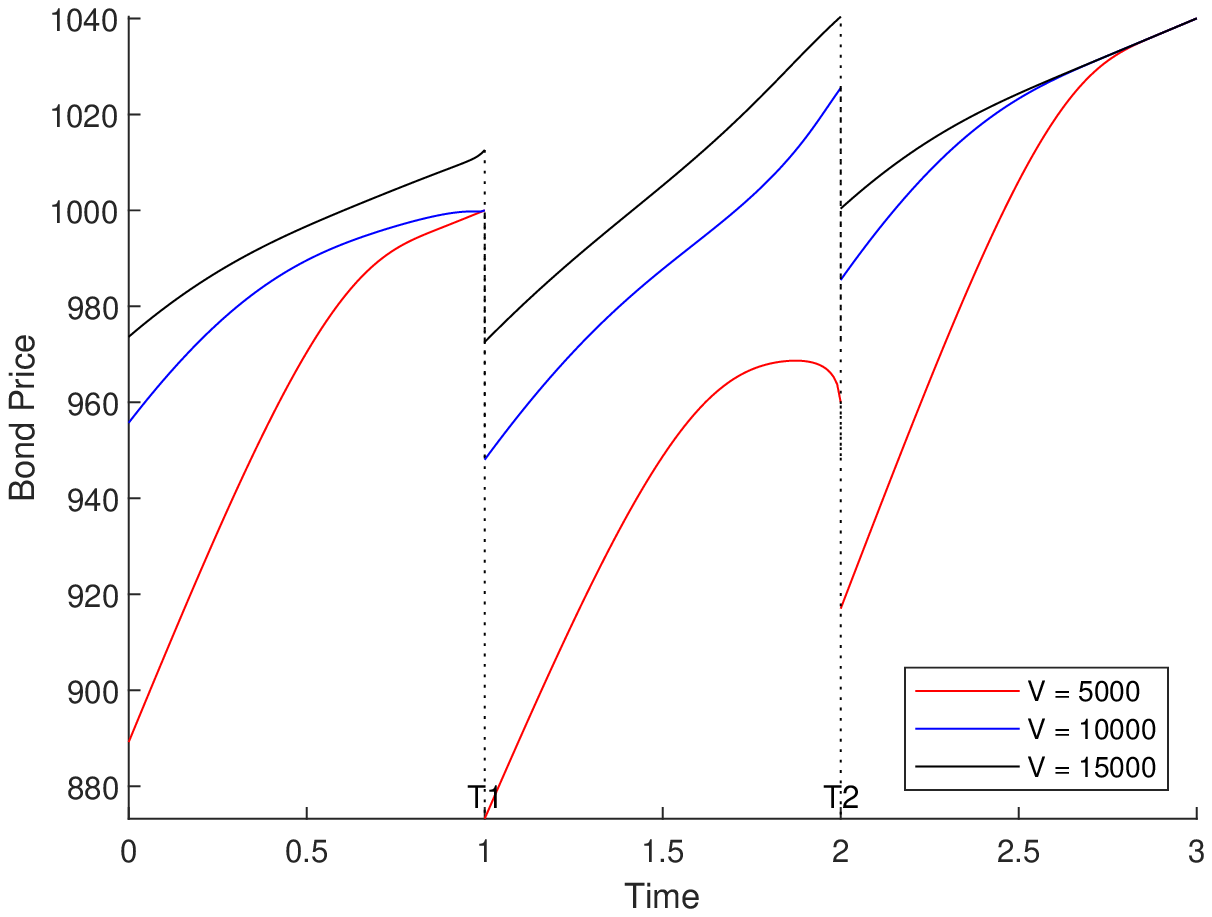}

Figure 9. $(t,\; B(V,\; t))$-graphs
\end{center}

Figure 6 shows the $(V,\; B_{0} (V,\; t))$-graphs at the times $t=0,\;0.25,\;0.5,\;0.75,\;1.0$, respectively when the firm value varies from 0 to 16000.

Figure 7 shows the $(V,\; B_{1} (V,\; t))$-graphs at the times $t=1.0,\;1.25,\;1.5,\;1.75,\;2.0$, respectively when the firm value varies from 0 to 16000.

Figure 8 shows the $(V,\; B_{2} (V,\; t))$-graphs at the times $t=2.0,\; 2.25,\; 2.5,\; 2.75,\; 3.0$, respectively when the firm value varies from 0 to 16000.

Figure 9 shows the $(t,\; B(V,\; t))$-graphs for firm values $V=5000,\; 10000,\; 15000$ on the time interval $[0,\; T]=[0,\; 3]$.

By the numerical calculation, $E_{1} =11945,\; \; E_{2} =5099,\; \; D_{1} =1000,\; \; D_{2} =960$.

1) We consider the case when $V=5000$ (red in Figure 9). Since $D_{1} <V<E_{1} $, the default event doesn't occur and the bond holder demands early redemption at $T_{1} $. Thus, as you can see in the red graph on $[0,T_{1}]$ in Figure 9, we get $B_{0} (V,\; T_{1} )=F$, and thus the bond does not exist on the interval $(T_{1} ,T_{3} ]$. So the real bond price is
\[B(V,\; t)=0,\; \; \; T_{1} <t\le T_{3} .\] 

On the other hand, the graphs of the interval $(T_{1} ,T_{2} ]$ in Figure 9 represent the bond price under the \textit{condition} that the default event didn't occur and the bond holder didn't demand early redemption at $T_{1} $. Under this assumption, since $D_{2} <V<E_{2} $, the default event doesn't occur and the bond holder demands early redemption at $T_{2} $. Thus, as you can see in the red graph on $(T_{1} ,T_{2} ]$ in Figure 9, we get $B_{1} (V,\; T_{2} )=F-C_{1} $ and thus the bond does not exist on the interval $(T_{2} ,T_{3} ]$.

And the graphs of the interval $(T_{2} ,T_{3} ]$ represent the bond price under the \textit{condition} that the default event didn't occur and the bond holder didn't demand early redemption at $T_{1} ,T_{2} $. Since $V>D_{3} =F+C_{3} $, the default event doesn't occur at $T_{3} $ and as you can see in the red graph on $(T_{2} ,T_{3} ]$ in Figure 9, $B(V,\; T_{3} )=F+C_{3} $.

2) We consider the case when $V=10000$ (blue in Figure 9). Since $D_{1} <V<E_{1} $ the default event doesn't occur and the bond holder demands early redemption at $T_{1} $. Thus, as you can see in the blue graph on $[0, T_{1}]$ in Figure 9, $B_{0} (V,\; T_{1} )=F$ and the bond does not exist on the interval $(T_{1} ,T_{3} ]$ like the first case.

On the other hand, the blue graphs of the interval $(T_{1} ,T_{2} ]$ in Figure 9 represent the bond price under the \textit{condition} that the default event didn't occur and the bond holder didn't demand early redemption $T_{1} $. Under this assumption, since $V>E_{2} >D_{2} $, the default event doesn't occur at $T_{2} $ and the bond holder keeps the contract. Thus, as you can see in the blue graph on $(T_{1} ,T_{2} ]$ in Figure 9, $B_{1} (V,\; T_{2} )=B_{2} (V,\; T_{2} )+C_{2} $. And the bond price on the interval $(T_{2} ,T_{3} ]$ is given by the blue curve on the last interval of Figure 9 and $B(V,\; T_{3} )=F+C_{3} $.

3) Next we consider the case when $V=15000$. Since $V>E_{1} >D_{1} $ the default event doesn't occur at $T_{1} $ and the bond holder keeps the contract. Thus, as you can see in the black graph on the first interval in Figure 9, $B_{0} (V,\; T_{1} )=B_{1} (V,\; T_{1} )+C_{1} $. The bond price on the interval $(T_{1} ,T_{2} ]$ is given by the black graph of the second interval of Figure 9 and $B_{1} (V,\; T_{2} )=B_{2} (V,\; T_{2} )+C_{2} $.And the bond price on $(T_{2},  T_{3} ]$ is given by the black graph of the last interval in Figure 9 and $B(V,\; T_{3} )=F+C_{3} $.

\begin{center}
\noindent \includegraphics*[width=14cm, height=10cm, keepaspectratio=false]{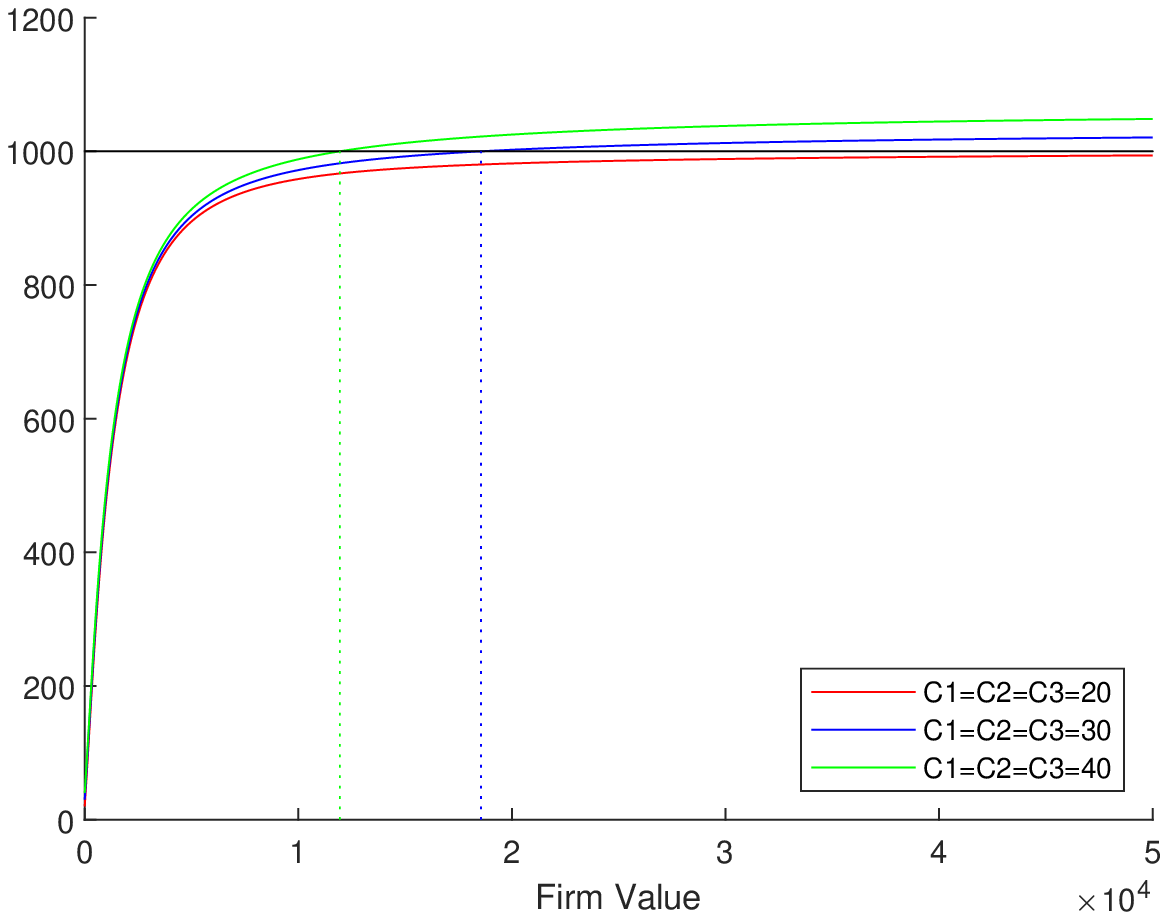}

Figure 10. The early redemption boundary at $T_{1} $.
\end{center}

\begin{center}
\includegraphics*[width=14cm, height=10cm, keepaspectratio=false]{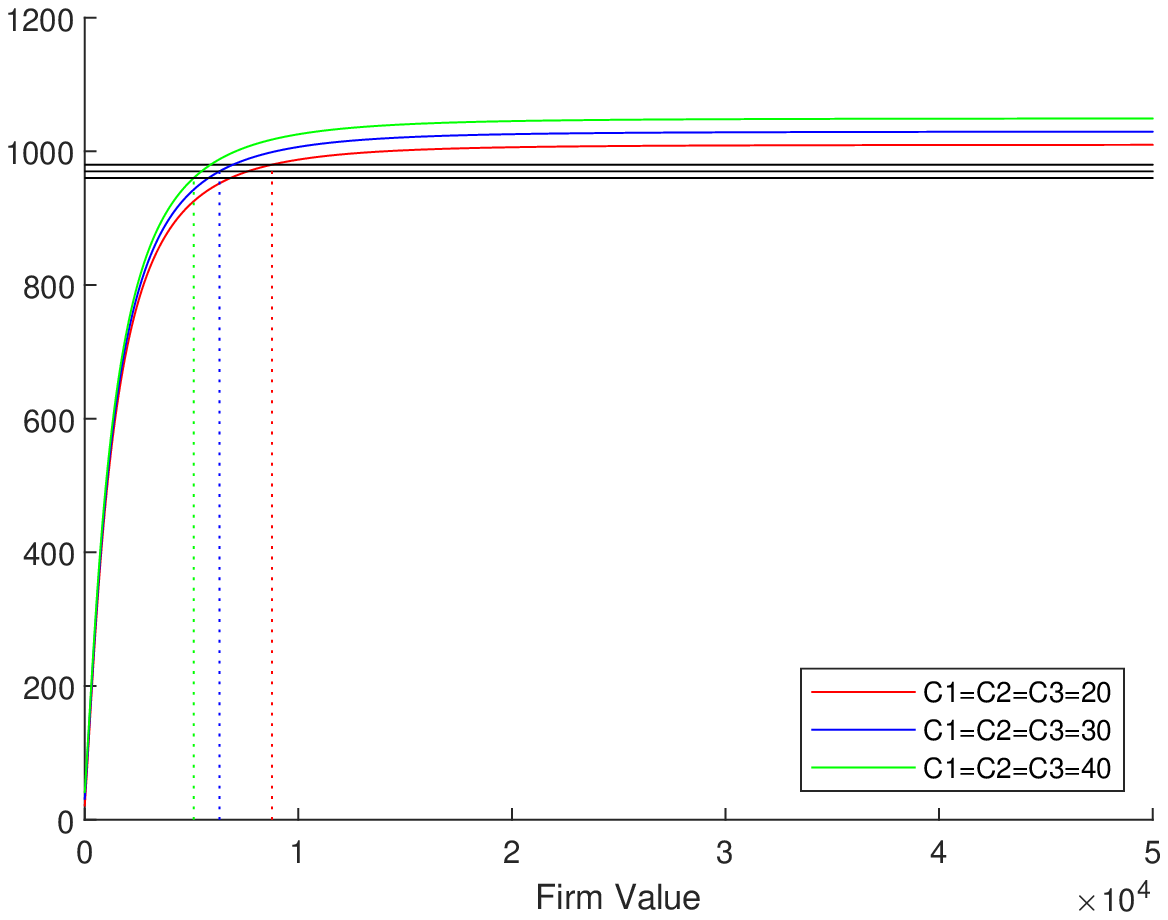}

Figure 11. The early redemption boundary at $T_{2} $.
\end{center}

Figure 10 and Figure 11 show the relation between coupons and early redemption boundaries.

As you can see in Figure 10, in the case when $C_{1} =C_{2} =C_{3} =20$, the assumption \eqref{eq15} (or \eqref{eq16}) is not satisfied and thus the bond holder should demand early redemption at $T_{1} $. In other cases, the assumption \eqref{eq15} is satisfied and thus the early redemption boundary exists. On the other hand, Figure 11 shows the early redemption boundary at $T_{2} $ under the \textit{condition} that the default event doesn't occur or the holder doesn't demand early redemption at $T_{1} $. The results show that increasing coupons makes the early redemption boundary smaller. This is compatible with their financial meaning.

\begin{center}
\includegraphics*[width=14cm, height=10cm, keepaspectratio=false]{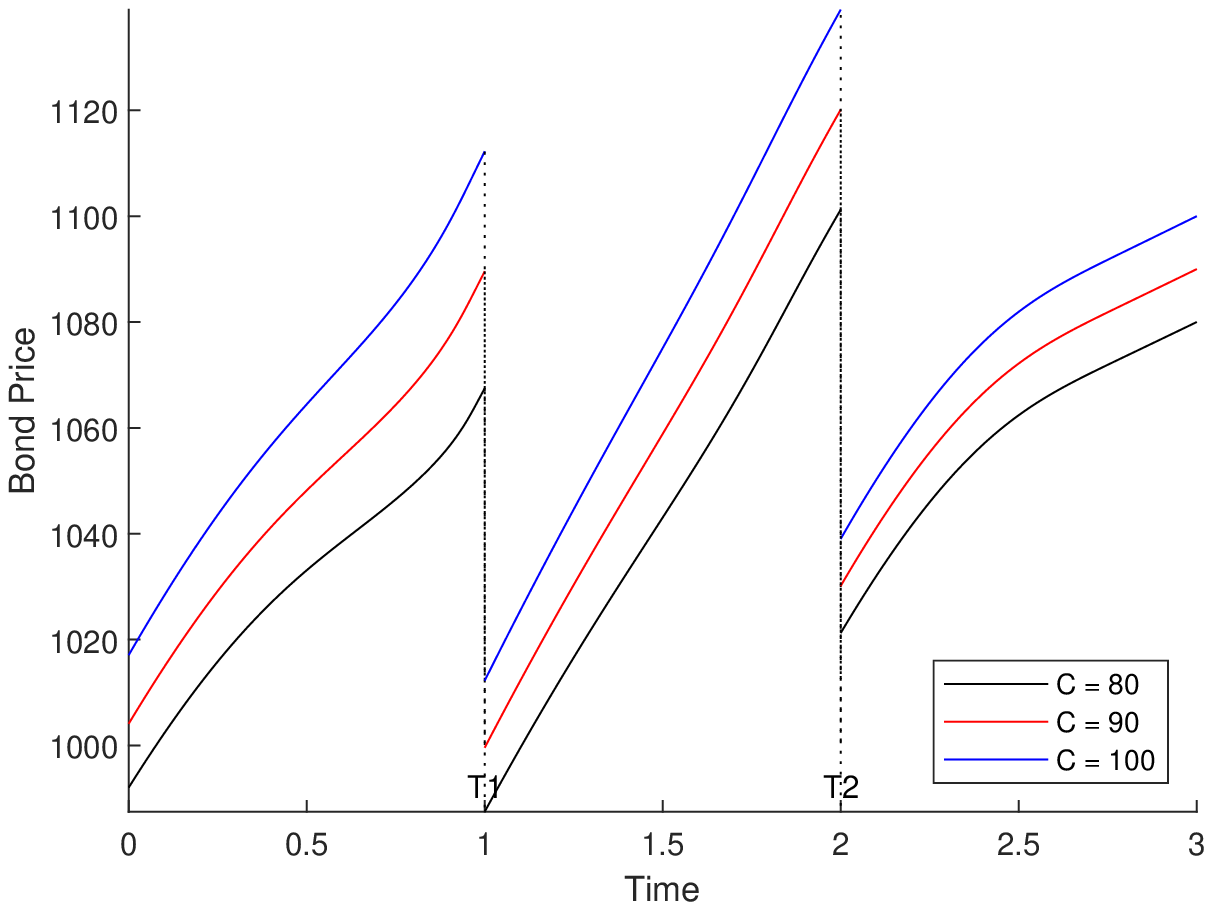}

Figure 12. The $(t,\; B(V,\; t))$-graphs for $V=10000$.
\end{center}

Figure 12 shows the effect of coupons on the bond price. As you can see in Figure 12, in the case of $C_{1} =C_{2} =C_{3} =80$, the initial bond price is less than the face value, in the case of $C_{1} =C_{2} =C_{3} =90$, the initial bond price is slightly larger than the face value, and in the case of $C_{1} =C_{2} =C_{3} =100$, the initial bond price is very larger than the face value. Thus if the coupons are too large, the firm must set up the initial bond price higher than the face value and if the firm wants to set up the initial bond price as the face value, the firm mustn't set up the coupon too large.

\section{Some applications of the pricing formula}

\subsection{Analysis on the duration of the bond}

A duration is a measure of average life of a bond\cite{Hull} and defined as follows:
\[D(V,\; t)=-\frac{1}{B(V,\; t)} \partial _{r} B(V,\; t;\; r).\] 

Now we will use the following notations:
\[f_{k} (r)=N_{k} (d_{1}^{-} ,\; \cdots ,\; d_{k}^{-} ;\; A_{k} ), g_{k} (r)=N_{k} (d_{1}^{+} ,\; \cdots ,d_{k-1}^{+} ,\; -\tilde{d}_{k}^{+} ;\; A_{k}^{-} ),\] 
\[h_{k} (r)=N_{k} (d_{1}^{-} ,\; \cdots ,\; d_{k-1}^{-} ,\; \tilde{\tilde{d}}_{k}^{-} ;\; A_{k} ).\] 
Then the initial bond price \eqref{eq31} can be written by
\begin{align}
B_{0} (V,\; 0)&=\sum _{k=0}^{N-1}\left[\bar{c}_{k+1} e^{-rT_{k+1} } f_{k+1} (r)\right] +e^{-bT_{k+1} } \delta V\sum_{k=0}^{N-1}g_{k+1} (r) \nonumber \\
&+\sum_{k=1}^{M}\left(F-\mathrm{\sum}_{j=1}^{k-1}\bar{c}_{j}\right)\cdot  1\{ D_{k} <E_{k} \} \cdot e^{-rT_{k} } [h_{k} (r)-f_{k} (r)].\nonumber
\end{align}
Thus we have
\begin{align}
-\partial _{r} B_{0}&=\sum _{k=0}^{N-1}\big[\bar{c}_{k+1} e^{-rT_{k+1} } (T_{k+1} f_{k+1} (r)-\partial _{r} f_{k+1})\big] -e^{-bT_{k+1} } \delta V\sum _{k=0}^{N-1}\partial _{r} g_{k+1} \nonumber \\
&+\sum _{k=1}^{M}(F-\Sigma_{j=1}^{k-1}\bar{c}_{j})\cdot  1\{ D_{k} <E_{k} \} \cdot e^{-rT_{k} } \big[T_{k} (h_{k}-f_{k} )-(\partial _{r} h_{k} -\partial _{r} f_{k})\big]. \label{eq32}
\end{align} 

The lemma on the derivative of multi-dimensional normal distribution function is as follows:

\begin{lemma} \label{lemma:6} \cite{OJKJ}
\[\partial _{x} N_{m} (a_{1} (x),\; \cdots ,\; a_{m} (x);\; A)=\sum _{i=1}^{m}\bar{N}_{m,i} (a_{1} (x),\; \cdots ,\; a_{m} (x);\; A)a'_{i} (x) \] 
is satisfied. Here
\[\bar{N}_{m,i} (a_{1} (x),\; \cdots ,\; a_{m} (x);\; A)=\int _{-\infty }^{a_{1} (x)}\cdots \int _{-\infty }^{a_{i-1} (x)}\int _{-\infty }^{a_{i+1} (x)}\cdots \int _{-\infty }^{a_{m} (x)}\frac{\sqrt{\det A} }{(\sqrt{2\pi } )^{m} } \exp \left(-\frac{1}{2} \stackrel{\frown}{y}_{i} (x)^{\bot } A\stackrel{\frown}{y}_{i} (x)\right)d\bar{y}_{i},\] 
\[\stackrel{\frown}{y}_{i} (x)^{\bot } =(y_{1} ,\; \cdots ,\; y_{i-1} ,\; a_{i} (x),\; y_{i+1} ,\; \cdots ,\; y_{m} ),\] 
\[d\bar{y}_{i} =dy_{1} \cdots dy_{i-1} dy_{i+1} \cdots dy_{m} ;\; \; \; i=1,\; \cdots ,\; m.\] 
\end{lemma}

Using Lemma \ref{lemma:6} and
\[\frac{\partial }{\partial r} d_{i}^{\pm } (0)=\frac{\partial }{\partial r} \tilde{d}_{i}^{\pm } (0)=\frac{\partial }{\partial r} \tilde{\tilde{d}}_{i}^{\pm } (0)=\frac{T_{i} }{s_{V} \sqrt{T_{i} } } ,\] 
we can get
\begin{align}
&\partial _{r} f_{k+1} (r)=\partial _{r} N_{k+1} (d_{1}^{-} ,\; \cdots ,\; d_{k+1}^{-} ;\; A_{k+1} )=\sum _{i=1}^{k+1}\bar{N}_{k+1,i} (d_{1}^{-} (r),\; \cdots ,\; d_{k+1}^{-} (r);\; A_{k+1} )\frac{T_{i} }{s_{V} \sqrt{T_{i} } } , \nonumber \\
&\partial _{r} g_{k+1} (r)=\partial _{r} N_{k+1} (d_{1}^{+} ,\; \cdots ,d_{k}^{+} ,\; -\tilde{d}_{k+1}^{+} ;\; A_{k+1}^{-} )= \nonumber \\
&=\sum _{i=1}^{k}\bar{N}_{k+1,i} (d_{1}^{-} (r),\; \cdots ,\; d_{k}^{+} (r),\; -\tilde{d}_{k+1}^{+} (r);\; A_{k+1} )\frac{T_{i} }{s_{V} \sqrt{T_{i}} }-\bar{N}_{k+1,k+1} (d_{1}^{-} (r),\; \cdots \; ,\; d_{k}^{+} (r),\; -\tilde{d}_{k+1}^{+} (r);\; A_{k+1} )\frac{T_{k+1} }{s_{V} \sqrt{T_{k+1}}}, \nonumber \\
&\partial _{r} h_{k} (r)=\partial _{r} N_{k} (d_{1}^{-} ,\; \cdots ,\; d_{k-1}^{-} ,\; \tilde{\tilde{d}}_{k}^{-} ;\; A_{k} )=\sum _{i=1}^{k}\bar{N}_{k,i} (d_{1}^{-} (r),\; \cdots ,\; d_{k-1}^{-} ,\; \tilde{\tilde{d}}_{k}^{-} (r);\; A_{k} )\frac{T_{i} }{s_{V} \sqrt{T_{i} } } . \label{eq33}
\end{align}

Substituting these expressions into \eqref{eq32}, then we have the following theorem:
\begin{theorem}[Duration] \label{theorem:4}
\begin{align}
\tilde{D}=\frac{-\partial _{r} B_{0} }{B_{0} }&=\frac{1}{B_{0} } \sum _{k=0}^{N-1}\big[\bar{c}_{k+1} e^{-rT_{k+1} }(T_{k+1} f_{k+1} (r)-F_{k+1})-e^{-bT_{k+1} } \delta V_{0} G_{k+1} \big]+\nonumber \\
&+\frac{1}{B_{0} } \sum _{k=1}^{M}(F-\Sigma_{j=1}^{k-1}\bar{c}_{j})\cdot  1\{ D_{k} <E_{k} \} \cdot e^{-rT_{k} } [T_{k} (h_{k} -f_{k} )-(H_{k} -F_{k})]. \nonumber
\end{align}
Here
\begin{align}
&F_{k} =\sum _{i=1}^{k}\bar{N}_{k,i} (d_{1}^{-} (r),\; \cdots ,\; d_{k}^{-} (r);\; A_{k} )\frac{T_{i} }{s_{V} \sqrt{T_{i} } }, \nonumber \\
&G_{k} =\sum _{i=1}^{k-1}\bar{N}_{k,i} (d_{1}^{-} (r),\; \cdots ,\; d_{k-1}^{-} (r),\; -\tilde{d}_{k}^{+} (r);\; A_{k} )\frac{T_{i} }{s_{V} \sqrt{T_{i} } }-\bar{N}_{k,k} (d_{1}^{-} (r),\; \cdots ,\; d_{k-1}^{-} (r),\; -\tilde{d}_{k}^{+} (r);\; A)\frac{T_{k} }{s_{V} \sqrt{T_{k} } } ,\nonumber \\
&H_{k} =\sum_{i=1}^{k}\bar{N}_{k,i} (d_{1}^{-} (r),\; \cdots ,\; d_{k-1}^{-} (r),\; \tilde{\tilde{d}}_{k}^{-} (r);\; A_{k} )\frac{T_{i} }{s_{V} \sqrt{T_{i} } }. \nonumber
\end{align}
\end{theorem}

\subsection{Credit Spread}
The credit spread is defined in every subintervals as follows:
\[CS_{i} =-\frac{\ln (B_{i} (V,\; t))-\ln Z_{i} (t;\; T)}{T-t} ,\; \; i=0,\; \cdots ,\; N-1.\] 
Here
\[Z_{i} (t;\; T)=\sum _{j=i+1}^{N}\bar{c}_{j} \cdot e^{-r(T_{j} -t)}  .\] 

In what follows, we give numerical examples of credit spread calculated by using Matlab. The basic data are as follows:
\[N=3,\; \; T_{1} =1,\; \; T_{2} =2,\; \; T_{3} =3(annum),\]
\[ r=0.03,\; \; b=0,\; \; s_{V} =1.0,\; \; \delta =0.5,\; \; F=1000,\; \; C_{1} =C_{2} =C_{3} =40.\] 
\begin{center}
\includegraphics*[width=14cm, height=5cm, keepaspectratio=false]{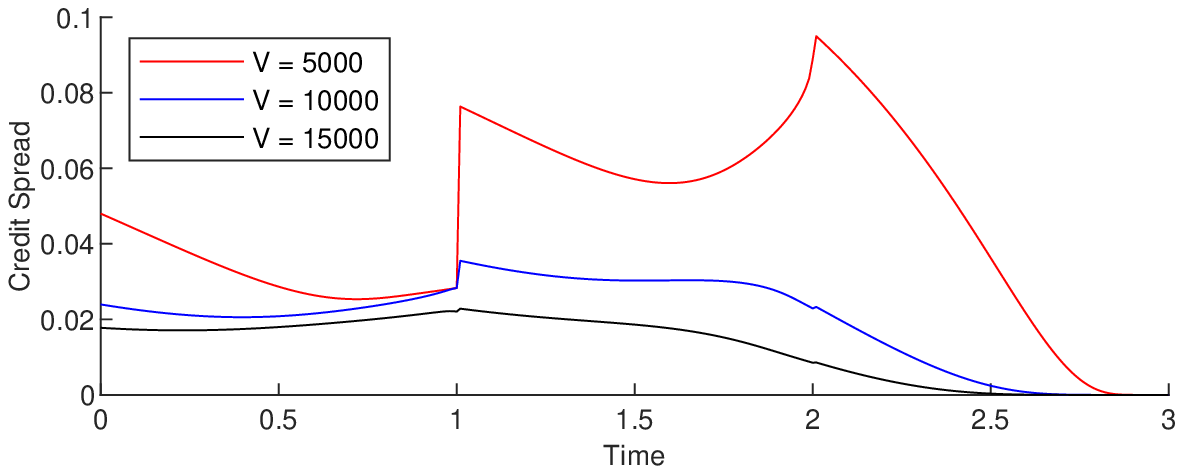}

Figure 13. $(t,\; CS)$-graphs for different firm values

\noindent \includegraphics*[width=14cm, height=5cm, keepaspectratio=false]{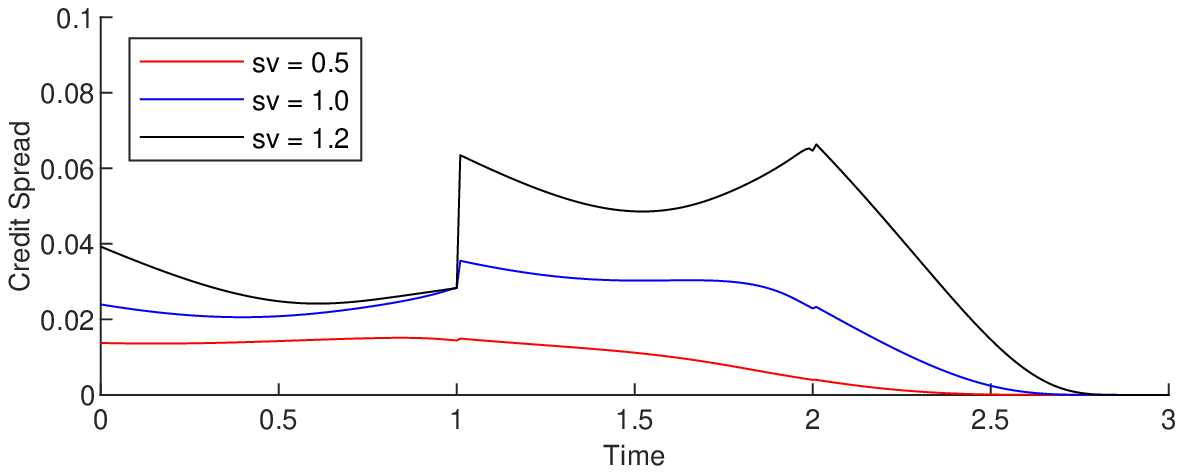}

Figure 14. $(t,\; CS)$-graphs for different volatilities

\noindent \includegraphics*[width=14cm, height=5cm, keepaspectratio=false]{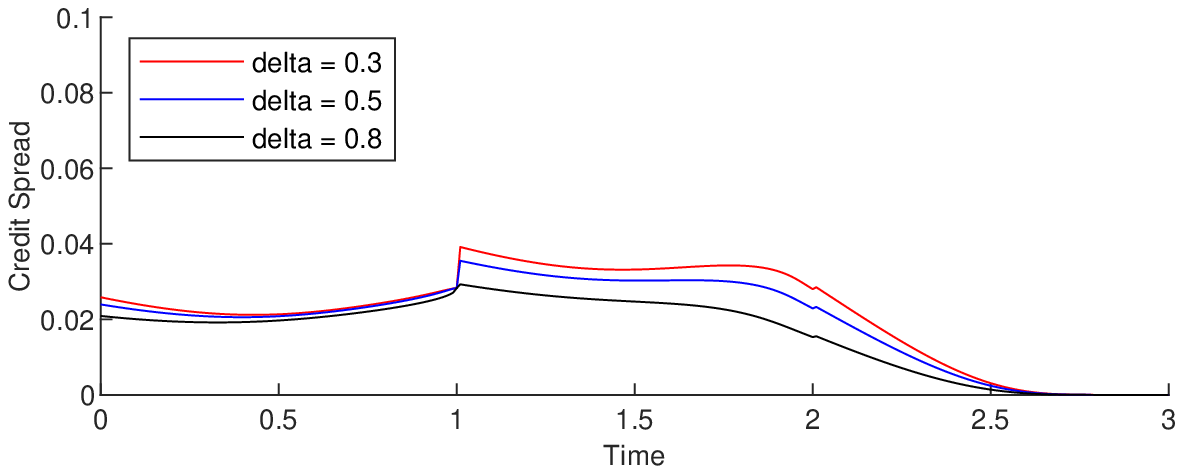}

Figure 15. $(t,\; CS)$-graphs for different recoveries

\noindent \includegraphics*[width=14cm, height=5cm, keepaspectratio=false]{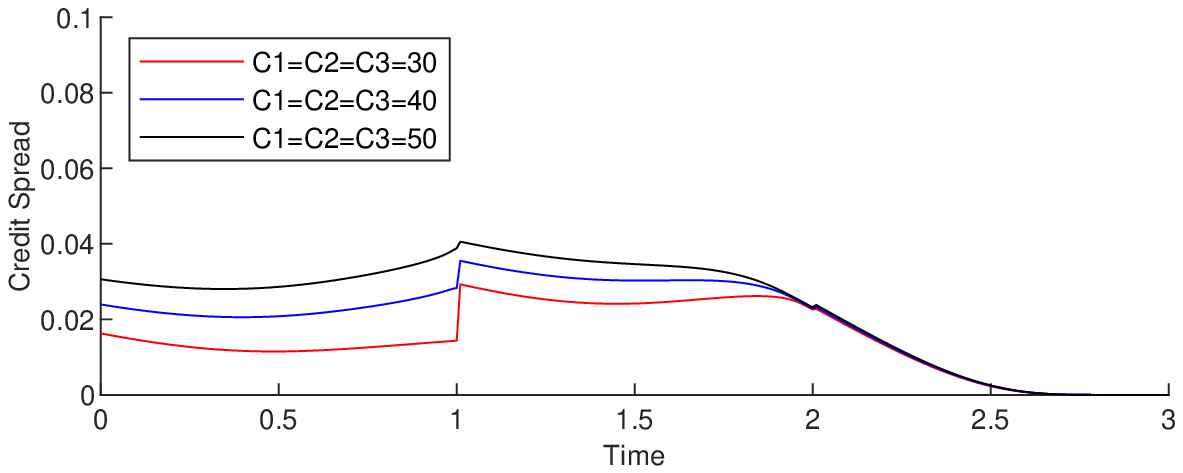}

Figure 16. $(t,\; CS)$-graphs for different coupons
\end{center}

Figure 13 shows the effect of firm value on credit spread. $(t,\; CS)$-graphs for firm values $V=5000,\; 10000,\; 15000$ are provided on time interval $[0,T]=[0,3]$. As you see in Figure 13, the credit spread is reduced when the firm value is increased. This is realated to the fact that the credit risk is reduced when the firm value becomes larger.

Figure 14 shows the effect of volatility on credit spread. $(t,\; CS)$-graphs for the volatilities $s_{V} =0.5,\; 1.0,\; 1.2$ are provided on time interval $[0,T]=[0,3]$. Here other parameters except for volatility are the same as the above and the firm value is fixed as $V=10000$. As you see in Figure 14, the credit spread is increased when the volatility becomes larger. This is realated to the fact that the credit risk becomes larger when the volatility becomes larger.

Figure 15 shows the effect of recovery rate on credit spread. $(t,\; CS)$-graphs for  recovery rates $\delta =0.3,\; 0.5,\; 0.8$ are provided on time interval $[0,T]=[0,3]$. Here other parameters except for recovery rate are the same as the above and the firm value is fixed as $V=10000$. As you see in Figure 15, the credit spread is reduced when the recovery rate is increased. This is realated to the fact that the bond price is increased when the recovery rate is increased.

Figure 16 shows the effect of coupon on credit spread.  $(t,\; CS)$-graphs for coupons $C=30,\; 40,\; 50$ are provided on time interval. Here other parameters except for coupon are the same as the above and the firm value is fixed as $V=10000$.  As you see in Figure 16, the credit spread is increased when the coupon is increased. This is realated to the fact that the credit risk becomes larger when the coupon is increased.
\section{Conclusion}
In this paper is derived the structural model of a discrete coupon bond with early redemption provision, using some analysis including min-max and gradient estimates of the bond price we studied the existence and uniqueness of default boundary and relationships between the design parameters of the bond, under some assumptions which are valid in finance we proved the existence and uniqueness of early redemption boundary and provided the analytic formula of the bond price by using higher binary options, and we gave the analysis on the duration and credit spread.

Our works provide some design guide for the discrete coupon bond with early redemption provision.

First, for the coupon bond with early redemption provision, coupons must be set up appropriately large such that \eqref{eq15} is satisfied, and if the firm does not want to pay so large coupon and thus \eqref{eq15} is not satisfied, then early redemption provision must be removed. In particular, if coupons are all the same, the face value and coupon must be set up such that \eqref{eq16} is satisfied.

\noindent Second, if the coupon is set up too large, early redemption is disadvantageous for the bond holder and the fair initial price of bond might be higher than the face value. In this case, if the firm sells the bond for the face value, then the firm may have a loss. Thus if the firm wants to sell the bond for the face value, the coupon must be set up not too large. If the firm set up the coupon too large, then the firm must set up the initial selling price of the bond higher than the face value.

\section{Appendix}
\textbf{Proof of Lemma 1.}
We use induction. First we consider the case when $i=N-1$. By the gradient estimate of Theorem \ref{theorem:1} and Corollary 2 of \cite{OJK}, we have
\[\min _{V} B_{N-1} (V,\; T_{N-1} )=B_{N-1} (0,\; T_{N-1} )=f_{N-1} (0)\cdot e^{-r(T_{N} -T_{N-1} )} .\]
On the other hand, from \eqref{eq8}, $f_{N-1} (0)=0$ and thus we have
\[\min _{V} B_{N-1} (V,\; T_{N-1} )=B_{N-1} (0,\; T_{N-1} )=0.\] 

Next in the case when $i=k$ we assume that
\[\min _{V} B_{k} (V,\; T_{k} )=B_{k} (0,\; T_{k} )=0\] 
and we will prove in the case when $i=k-1$. By the gradient estimate of Theorem \ref{theorem:1} and Corollary 2 of [20], we have
\[\min _{V} B_{k-1} (V,\; T_{k-1} )=B_{k-1} (0,\; T_{k-1} )=f_{k-1} (0)\cdot e^{-r(T_{k} -T_{k-1} )} .\]
Using induction assumption, we have
\[\max \left\{[B_{k} (V,\; T_{k} )+\bar{c}_{k} ],\; \; F-\Sigma _{j=1}^{k-1}\bar{c}_{j}  \right\}\ge [B_{k} (V,\; T_{k} )+\bar{c}_{k} ]\ge \bar{c}_{k} >0.\] 
Then from \eqref{eq9} $f_{k-1} (0)=\delta \cdot 0=0$ and thus we have
\[\min _{V} B_{k-1} (V,\; T_{k-1} )=B_{k-1} (0,\; T_{k-1} )=0.\]    $\Box$

\vspace{3mm}
\textbf{Proof of Lemma 2.} By the gradient estimate of Theorem \ref{theorem:1} and Corollary 2 of [20], we have
\[\sup _{V} \left[B_{i-1} (V,\; T_{i-1} )\right]=B_{i-1} (+\infty ,\; T_{i-1} )=f_{i-1} (+\infty )\cdot e^{-r(T_{i} -T_{i-1} )} .\] 
From \eqref{eq8} and \eqref{eq9}, we have
\[f_{i-1} (+\infty )=\; \max \left\{B_{i} (+\infty ,\; T_{i} )+\bar{c}_{i},\; F-\Sigma _{j=1}^{i-1}\bar{c}_{j}  \right\}\] 
and from \eqref{eq13}, we have
\[\sup _{V} \left[B_{i-1} (V,\; T_{i-1} )\right]\; =(F-\Sigma _{j=1}^{i-1}\bar{c}_{j}  )\cdot e^{-r(T_{i} -T_{i-1} )} .\] 
On the other hand,
\begin{align}
\sup _{V} [B_{i-1} (V,\; T_{i-1} )+\bar{c}_{i-1} ]-(F-\Sigma _{j=1}^{i-2}\bar{c}_{j}  )&=[(F-\Sigma_{j=1}^{i-1}\bar{c}_{j})\cdot e^{-r(T_{i} -T_{i-1} )} +\bar{c}_{i-1} ]-(F-\Sigma_{j=1}^{i-2}\bar{c}_{j}) \nonumber \\
&=(F-\Sigma_{j=1}^{i-1}\bar{c}_{j})[e^{-r(T_{i} -T_{i-1} )} -1].\nonumber
\end{align}
From Lemma \ref{lemma:1} and \eqref{eq13} we have $F-\sum _{j=1}^{i-1}\bar{c}_{j}  \ge 0$ and thus we have
\[\sup _{V} [B_{i-1} (V,\; T_{i-1} )+\bar{c}_{i-1} ]-(F-\Sigma_{j=1}^{i-2}\bar{c}_{j}  )\le 0.\] 
Therefore, the lemma is proved.   $\Box$

\vspace{3mm}
\textbf{Proof of Lemma 3.} We use induction. First we consider the case when $i=N-$ $1$. From the gradient estimate of Theorem \ref{theorem:1} and Corollary 2 of [20], we have
\[\sup _{V} \left[B_{N-1} (V,\; T_{N-1} )\right]=B_{N-1} (+\infty ,\; T_{k-1} )=f_{N-1} (+\infty )\cdot e^{-r(T_{N} -T_{N-1} )} .\] 
From \eqref{eq8} we have $f_{N-1} (+\infty )=\bar{c}_{N} $ and thus we have
\[B_{N-1} (+\infty ,\; T_{N-1} )=\bar{c}_{N} \cdot e^{-r(T_{N} -T_{N-1} )} .\] 
Thus the assertion of our lemma holds in this case.

Next in the case when $i=k$ we assume that
\[\sup _{V} B_{k} (V,\; T_{k} )=B_{k} (+\infty ,\; T_{k} )=\sum _{j=k+1}^{N}\left[\bar{c}_{j} e^{-r(T_{j} -T_{k} )} \right] \] 
and we will prove in the case when $i=k-1$. From the gradient estimate of Theorem \ref{theorem:1} and Corollary 2 of [20], we have
\[\sup _{V} \left[B_{k-1} (V,\; T_{k-1} )\right]=B_{k-1} (+\infty ,\; T_{k-1} )=f_{k-1} (+\infty )\cdot e^{-r(T_{k} -T_{k-1} )} .\] 
And from \eqref{eq9} we have
\[f_{k-1} (+\infty )=\max \left\{B_{k} (+\infty ,\; T_{k} )+\bar{c}_{k} {\rm ,}\; F-\Sigma_{j=1}^{k-1}\bar{c}_{j}  \right\}.\] 
Thus from Corollary \ref{corollary:1}, we have
\[\sup _{V} \left[B_{k-1} (V,\; T_{k-1} )\right]=[B_{k} (+\infty ,\; T_{k} )+\bar{c}_{k} ]\cdot e^{-r(T_{k} -T_{k-1} )} .\] 
Substituting the induction assumption of the case when $i=k$ into the above expression, we have
\[\sup _{V} [B_{k-1} (V,\; T_{k-1} )]=\left\{\bar{c}_{k} +\sum _{j=k+1}^{N}\left[\bar{c}_{j} e^{-r(T_{j} -T_{k} )} \right] \; \right\}\cdot e^{-r(T_{k} -T_{k-1} )} =\sum _{j=k}^{N}\left[\bar{c}_{j} e^{-r(T_{j} -T_{k-1} )} \right] \; .\] 
Lemma is proved.                 $\Box$

\vspace{3mm}
\textbf{Proof of Lemma 4.} From the assumption
\begin{align}
\sum _{j=m+1}^{N}\left[\bar{c}_{j} e^{r(T_{N} -T_{j} )} \right]&>(F-\Sigma_{j=1}^{m-1}\bar{c}_{j})\cdot e^{r(T_{N} -T_{m})} -\bar{c}_{m} e^{r(T_{N} -T_{m} )} \nonumber \\
&=(F-\Sigma_{j=1}^{m}\bar{c}_{j})\cdot e^{r(T_{N} -T_{m} )} >(F-\Sigma_{j=1}^{m}\bar{c}_{j})\cdot e^{r(T_{N} -T_{m-1} )}. \nonumber
\end{align}
Lemma is proved.   $\Box$

\vspace{3mm}
\textbf{Proof of Lemma 5.} The proof in the case when $i=N-1$ is the same as in Lemma \ref{lemma:3}.

Next in the case when $i=k$, we assume that
\[\sup _{V} B_{k} (V,\; T_{k} )=B_{k} (+\infty ,\; T_{k} )=\sum _{j=k+1}^{N}\left[\bar{c}_{j} e^{-r(T_{j} -T_{k} )} \right] \] 
and we will prove in the case when $i=k-1$. From the gradient estimate of Theorem \ref{theorem:1} and the Corollary 2 of [20], we have

\begin{equation}
\sup _{V} \left[B_{k-1} (V,\; T_{k-1} )\right]=B_{k-1} (+\infty ,\; T_{k-1} )=f_{k-1} (+\infty )\cdot e^{-r(T_{k} -T_{k-1} )}.  \tag{A.1} \label{A.1}
\end{equation}

\noindent On the other hand, from \eqref{eq9} we have
\[f_{k-1} (+\infty )=\max \left\{B_{k} (+\infty ,\; T_{k} )+\bar{c}_{k} {\rm ,}\; F-\Sigma_{j=1}^{k-1}\bar{c}_{j}  \right\}.\]
Substituting the above expression into \eqref{A.1} and using the induction assumption of the case of $i=k$ , we have
\begin{align}
\sup _{V} [B_{k-1} (V,\; T_{k-1} )]&=\max \left\{\sum_{j=k+1}^{N}\left[\bar{c}_{j} e^{-r(T_{j} -T_{k} )} \right] +\bar{c}_{k} ,\; F-\sum _{j=1}^{k-1}\bar{c}_{j}  \right\}\cdot e^{-r(T_{k} -T_{k-1} )} \nonumber \\
&=\max \left\{\sum _{j=k}^{N}\left[\bar{c}_{j} e^{-r(T_{j} -T_{k-1} )} \right] \; ,\; \left(F-\sum _{j=1}^{k-1}\bar{c}_{j}  \right)\cdot e^{-r(T_{k} -T_{k-1} )} \right\}\;. \nonumber
\end{align}
From Corollary \ref{corollary:2} we have
\[\sup _{V} \left[B_{k-1} (V,\; T_{k-1} )\right]=B_{k-1} (+\infty ,\; T_{k-1} )\; =\sum _{j=k}^{N}\left[\bar{c}_{j} e^{-r(T_{j} -T_{k-1} )} \right] .\]   $\Box$

\textbf{Proof of Corollary 3.} The necessity has already been mentioned in Remark \ref{remark:2}. We prove the sufficiency. From Lemma \ref{lemma:5}, \eqref{eq15} implies
\begin{equation}
\sup _{V} B_{1} (V,\; T_{1} )=\sum _{j=2}^{N}\left[\bar{c}_{j} e^{-r(T_{j} -T_{1} )} \right].                \tag{A.2} \label{A.2}
\end{equation}

\noindent Multiplying $e^{-r(T_{N} -T_{1} )} $ to the both sides of \eqref{eq15} and then rewriting it, \eqref{eq15} is equivalent to $\sum _{j=1}^{N}[\bar{c}_{j} e^{-r(T_{j} -T_{1} )}] >F$, and considering \eqref{A.2}, we have

\[\sup _{V} B_{1} (V,\; T_{1} )+\bar{c}_{1} >F.\]
$\Box$

\vspace{3mm}
\textbf{Proof of Theorem 3.} First we consider the case when $M\le i\le N-1$. In this case $B_{i} (V,\; t)$ is given by \eqref{eq22}, i.e.,
\begin{align}
B_{i} (V,\; t)&=\sum _{k=i}^{N-1}\Big[\bar{c}_{k+1} B_{D_{i+1} \cdots D_{k} D_{k+1} }^{\; +\; \; \; \cdots \; +\; \; +} (V,\; t;\; T_{i+1} ,\; \cdots ,\; T_{m+1} ;\; r,\; b,\; s_{V} )\nonumber\\
&+\delta A_{D_{i+1} \; \cdots D_{k} D_{k+1} }^{\; +\; \; \; \cdots \; \; +\; \;  -} (V,\; t;\; T_{i+1} ,\; \cdots ,\; T_{m} ,\; T_{m+1} ;\; r,\; b,\; s_{V} )\Big]. \nonumber
\end{align}
Now considering \eqref{eq27}, then the above expression can be written into
\begin{align}
B_{i} (V,\; t)&=\sum _{k=i}^{N-1}\Big[\bar{c}_{k+1} B_{U_{i+1} \cdots U_{k} U_{k+1} }^{\; +\; \; \; \cdots \; +\; +} (V,\; t;\; T_{i+1} ,\; \cdots ,\; T_{m+1} ;\; r,\; b,\; s_{V} ) \nonumber \\
&+\delta A_{U_{i+1} \; \cdots U_{k} D_{k+1} }^{\; +\; \; \; \cdots \; \; +\; \; -} (V,\; t;\; T_{i+1} ,\; \cdots ,\; T_{m} ,\; T_{m+1} ;\; r,\; b,\; s_{V} )\Big]. \nonumber
\end{align}
This can be rewritten as follows:
\begin{align}
B_{i} (V,\; t)&=\sum _{k=i}^{N-1}\big[\bar{c}_{k+1} B_{U_{i+1} \cdots U_{k} U_{k+1} }^{\; +\; \cdots \;\; +\; \; \; +} (V,\; t;\; T_{i+1} ,\; \cdots ,\; T_{k+1} ;\; r,\; b,\; s_{V})+ \nonumber \\
&+\delta A_{U_{i+1} \; \cdots U_{k} D_{k+1} }^{\; +\; \; \cdots \;\; +\; \; -} (V,\; t;\; T_{i+1} ,\; \cdots ,\; T_{k} ,\; T_{k+1} ;\; r,\; b,\; s_{V} )\big]+ \nonumber \\
&+\sum _{k=i+1}^{M}(F-\Sigma _{j=1}^{k-1}\bar{c}_{j})\cdot  1\{ D_{k} <E_{k} \} \cdot \big[B_{U_{i+1} \cdots U_{k-1} L_{k} }^{\; +\; \; \cdots \; \; +\; \; \; \; +} (V,\; t;\; T_{i+1} ,\; \cdots ,\; T_{k} ;\; r,\; b,\; s_{V} )- \nonumber\\
&-B_{U_{i+1} \cdots U_{k-1} U_{k} }^{\; +\; \; \cdots \; \; +\; \; \; \; +} (V,\; t;\; T_{i+1} ,\; \cdots ,\; T_{k} ;\; r,\; b,\; s_{V} )\big]\; ,\; \; \; T_{i} <t<T_{i+1}. \nonumber
\end{align}
Indeed, $i+1>M$and thus in the above expesssion, the term of second sum does not exist. Thus we have \eqref{eq28}.

Next we consider the case when $0\le i\le M-1$. We use induction. First we prove in the case when $i=M-1$. From \eqref{eq26}, we have
\begin{align}
B_{M-1} (V,\; t)&=\sum _{k=M-1}^{N-1}\big[\bar{c}_{k+1} B_{U_{M} \; D_{M+1} \cdots D_{k} D_{k+1} }^{\;+\; \; \; \; +\; \; \; \cdots \; \; +\; \;  +} (V,\; t;\; T_{M} ,\; \cdots ,\; T_{k+1} ;\; r,\; b,\; s_{V} ) \nonumber \\
&+\delta A_{U_{M} \; D_{M+1} \; \cdots D_{k} D_{k+1} }^{\;+\;  \; \; \; +\; \; \; \; \; \cdots \; +\; \;  -} (V,\; t;\; T_{M} ,\; T_{M+1} ,\; \cdots ,\; T_{k+1} ;\; r,\; b,\; s_{V} )\big]+ \nonumber \\
&+(F-\Sigma_{j=1}^{M-1}\bar{c}_{j})\cdot 1\{ D_{M} <E_{M} \} \cdot B_{L_{M} }^{+} (V,\; t;\; T_{M} ;\; r,\; b,\; s_{V} )- \nonumber \\
&-(F-\Sigma _{j=1}^{M-1}\bar{c}_{j})\cdot 1\{ D_{M} <E_{M} \} \cdot B_{U_{M} }^{+} \cdot (V,\; t;\; T_{M} ;\; r,\; b,\; s_{V} ). \nonumber
\end{align}
Recalling \eqref{eq27}, this can be rewritten as follows:
\begin{align}
B_{M-1} (V,\; t)&=\sum _{k=M-1}^{N-1}\big[\bar{c}_{k+1} B_{U_{M} U_{M+1} \cdots U_{k} U_{k+1} }^{\; +\; \; +\; \; \; \; \; \cdots \; +\; \; +} (V,\; t;\; T_{M} ,\; \cdots ,\; T_{k+1} ;\; r,\; b,\; s_{V} ) \nonumber \\
&+\delta A_{U_{M} U_{M+1} \cdots U_{k} D_{k+1} }^{\; +\; \; \; +\; \; \; \; \cdots \; +\; \; -} (V,\; t;\; T_{M} ,\; T_{M+1} ,\; \cdots ,\; T_{k+1} ;\; r,\; b,\; s_{V} )\big] \nonumber \\
&+(F-\Sigma _{j=1}^{M-1}\bar{c}_{j}  )\cdot 1\{ D_{M} <E_{M} \} \cdot B_{L_{M} }^{+} (V,\; t;\; T_{M} ;\; r,\; b,\; s_{V} )- \nonumber \\
&-(F-\Sigma _{j=1}^{M-1}\bar{c}_{j}  )\cdot 1\{ D_{M} <E_{M} \} \cdot B_{U_{M} }^{+} \cdot (V,\; t;\; T_{M} ;\; r,\; b,\; s_{V} ). \nonumber
\end{align}
Thus we obtained \eqref{eq28}.

Next in the case when $i=m$, we assume that 

\begin{align}
B_{m} (V,\; t)&=\sum _{k=m}^{N-1}\big[\bar{c}_{k+1} B_{U_{m+1} \cdots U_{k} U_{k+1} }^{\; +\; \; \cdots \;\; +\; \; +} (V,\; t;\; T_{m+1} ,\; \cdots ,\; T_{k+1} ;\; r,\; b,\; s_{V})+ \nonumber \\
&+\delta A_{U_{m+1} \; \cdots U_{k} D_{k+1} }^{\; +\; \; \; \cdots \;\; +\; \; -} (V,\; t;\; T_{m+1} ,\; \cdots ,\; T_{k} ,\; T_{k+1} ;\; r,\; b,\; s_{V} )\big]+ \tag{A.3} \label{A.3} \\
&+\sum _{k=m+1}^{M}(F-\Sigma _{j=1}^{k-1}\bar{c}_{j}  )\cdot  1\{ D_{k} <E_{k} \} \cdot \big[B_{U_{m+1} \cdots U_{k-1} L_{k} }^{\; +\; \; \cdots \; \; +\; \; \; \; \; +} (V,\; t;\; T_{m+1} ,\; \cdots ,\; T_{k} ;\; r,\; b,\; s_{V})- \nonumber \\
&-B_{U_{m+1} \cdots U_{k-1} U_{k} }^{\; +\; \; \cdots \; \; +\; \; \; \;\; +} (V,\; t;\; T_{m+1} ,\; \cdots ,\; T_{k} ;\; r,\; b,\; s_{V} )\big]\; ,\; \; \; T_{m} <t<T_{m+1}. \nonumber
\end{align}

\noindent and then we prove in the case when $i=m-1$. From \eqref{eq25}, 
\begin{align}
B_{m-1} (V,\; T_{m} )&=[B_{m} (V,\; T_{m} )+\bar{c}_{m} ]\cdot 1\{ V\ge U_{m} \} +(F-\Sigma _{j=1}^{m-1}\bar{c}_{j}  )\cdot 1\{ D_{m} <E_{m} \} \cdot 1\{ V\ge L_{m} \} - \nonumber\\
&-(F-\Sigma_{j=1}^{m-1}\bar{c}_{j}  )\cdot 1\{ D_{m} <E_{m} \} \cdot 1\{ V\ge U_{m} \} +\delta V\cdot 1\{ V<D_{m} \}. \nonumber
\end{align}
Substituting \eqref{A.3} into the above expression, we have
\begin{align}
B_{m-1} (V,\; T_{m})&=\bigg\{\sum _{k=m}^{N-1}\big[\bar{c}_{k+1} B_{U_{m+1} \cdots U_{k} U_{k+1} }^{\; +\; \; \cdots \;\; +\; \; +} (V,\; T_{m} ;\; T_{m+1} ,\; \cdots ,\; T_{k+1} ;\; r,\; b,\; s_{V} )+ \nonumber \\
&+\delta A_{U_{m+1} \; \cdots U_{k} D_{k+1} }^{\; +\; \; \; \cdots \;\; +\; \; -} (V,\; T_{m} ;\; T_{m+1} ,\; \cdots ,\; T_{k} ,\; T_{k+1} ;\; r,\; b,\; s_{V} )\big]+ \nonumber \\
&+\sum _{k=m+1}^{M}(F-\Sigma _{j=1}^{k-1}\bar{c}_{j}  )\cdot  1\{ D_{k} <E_{k} \} \cdot \big[B_{U_{m+1} \cdots U_{k-1} L_{k} }^{\; +\; \; \cdots \; \; +\; \; \; \;\; +} (V,\; T_{m} ;\; T_{m+1} ,\; \cdots ,\; T_{k} ;\; r,\; b,\; s_{V})- \nonumber \\
&-B_{U_{m+1} \cdots U_{k-1} U_{k} }^{\; +\; \; \cdots \; \; +\; \; \; \; \;+} (V,\; T_{m} ;\; T_{m+1} ,\; \cdots ,\; T_{k} ;\; r,\; b,\; s_{V} )\big]\bigg\}\cdot 1\{ V\ge U_{m} \} + \nonumber \\
&+(F-\Sigma _{j=1}^{m-1}\bar{c}_{j}  )\cdot 1\{ D_{m} <E_{m} \} \cdot 1\{ V\ge L_{m} \} -(F-\Sigma _{j=1}^{m-1}\bar{c}_{j}  )\cdot 1\{ D_{m} <E_{m} \} \cdot 1\{ V\ge U_{m} \}+ \nonumber \\
&+\bar{c}_{m} \cdot 1\{ V\ge U_{m} \} +\delta V\cdot 1\{ V<D_{m} \} . \nonumber
\end{align}
Using the pricing formula of higher order binary option \cite{OK}, we have
\begin{align}
B_{m-1} (V,\; t)&=\sum _{k=m}^{N-1}\big[\bar{c}_{k+1} B_{U_{m} U_{m+1} \cdots U_{k} U_{k+1} }^{\;+\; \; +\; \; \; \; \; \cdots  +\; \; +} (V,\; t;\; T_{m} ,\; T_{m+1} ,\; \cdots ,\; T_{k+1} ;\; r,\; b,\; s_{V} )+ \nonumber \\
&+\delta A_{U_{m} U_{m+1} \; \cdots U_{k} D_{k+1} }^{\;+ \; +\; \; \; \; \; \; \cdots \; +\; \; -} (V,\; t;\; T_{m} ,\; T_{m+1} ,\; \cdots ,\; T_{k} ,\; T_{k+1} ;\; r,\; b,\; s_{V} )\big]+ \nonumber \\
&+\sum _{k=m+1}^{M}(F-\Sigma _{j=1}^{k-1}\bar{c}_{j}  )\cdot  1\{ D_{k} <E_{k} \} \cdot \big[B_{U_{m} U_{m+1} \cdots U_{k-1} L_{k} }^{\;+ \; \; +\; \; \; \; \; \cdots \; +\; \; \;\; -} (V,\; t;\; T_{m} ,\; T_{m+1} ,\; \cdots ,\; T_{k} ;\; r,\; b,\; s_{V} )- \nonumber \\
&-B_{U_{m} U_{m+1} \cdots U_{k-1} U_{k} }^{\;+\; \; +\; \; \; \; \; \cdots \; +\; \; \; \;-} (V,\; t;\; T_{m} ,\; T_{m+1} ,\; \cdots ,\; T_{k} ;\; r,\; b,\; s_{V} )\big]+ \nonumber \\
&+(F-\Sigma _{j=1}^{m-1}\bar{c}_{j}  )\cdot 1\{ D_{m} <E_{m} \} \cdot B_{L_{m} }^{+} (V,\; t;\; T_{m} ;\; r,\; b,\; s_{V} )- \nonumber \\
&-(F-\Sigma _{j=1}^{m-1}\bar{c}_{j}  )\cdot 1\{ D_{m} <E_{m} \} \cdot B_{U_{m} }^{+} \cdot (V,\; t;\; T_{m} ;\; r,\; b,\; s_{V} )+ \nonumber \\
&+\bar{c}_{m} B_{U_{m} }^{+} (V,\; t;\; T_{m} ;\; r,\; b,\; s_{V} )+\delta A_{D_{m} }^{-} \cdot (V,\; t;\; T_{m} ;\; r,\; b,\; s_{V} ). \nonumber
\end{align}
Here if we make the last 4 terms outside of $\sum $ put into the corresponding sum terms respectively, we have
\begin{align}
B_{m-1} (V,\; t)&=\sum _{k=m-1}^{N-1}\big[\bar{c}_{k+1} B_{U_{m} \cdots U_{M} D_{M+1} \cdots D_{k+1} }^{\;  +\; \cdots \;  +\; \; \; +\; \; \; \;\; \cdots\;  +} (V,\; t;\; T_{m} ,\; \cdots ,\; T_{k+1} ;\; r,\; b,\; s_{V})+ \nonumber \\
&+\delta A_{U_{m} \; \cdots U_{M} D_{M+1} \cdots D_{k+1} }^{ \; +\; \; \cdots \; +\; \; \; +\; \; \;\;\; \cdots \;  -} (V,\; t;\; T_{m} ,\; \cdots ,\; T_{k} ,\; T_{k+1} ;\; r,\; b,\; s_{V} )\big]+ \nonumber \\
&+\sum _{k=m}^{M}(F-\Sigma _{j=1}^{k-1}\bar{c}_{j}  )\cdot  1\{ D_{k} <E_{k} \} \cdot \big[B_{U_{m} \cdots \; U_{k-1} L_{k} }^{\; +\; \; \cdots \; +\; \; \;\; +} (V,\; t;\; T_{m} ,\; \cdots ,\; T_{k} ;\; r,\; b,\; s_{V})- \nonumber \\
&-B_{U_{m} \cdots \; U_{k-1} U_{k} }^{\; +\; \; \cdots \; +\; \; \;\;\; +} (V,\; t;\; T_{m} ,\; \cdots ,\; T_{k} ;\; r,\; b,\; s_{V} )\big]. \nonumber
\end{align}
$\Box$

\pagebreak

\end{document}